\documentclass[a4paper,11pt]{article}

\usepackage[margin=1in]{geometry}

\usepackage[utf8x]{inputenc}
\usepackage{amsthm}
\usepackage{algorithmic}
\usepackage{graphicx}
\usepackage{algorithm}
\usepackage{color}

\usepackage{amssymb,amsmath}
\usepackage{amsopn}
\usepackage{mathabx}
\usepackage{xspace}
\usepackage{hyperref}

\newcommand{\R}{\mathbb{R}}

\newcommand{\N}{\mathbb{N}}

\newcommand\cG{{\ensuremath{\mathcal{G}}}\xspace}
\newcommand\cH{{\ensuremath{\mathcal{H}}}\xspace}

\newcommand\cK{{\ensuremath{\mathcal{K}}}\xspace}

\usepackage{mathtools}

\newtheorem{theorem}{Theorem}
\newtheorem{lemma}[theorem]{Lemma}
\newtheorem{proposition}[theorem]{Proposition}
\newtheorem{corollary}[theorem]{Corollary}
\newtheorem{definition}{Definition}
\newtheorem{remark}[theorem]{Remark}

\title{SoS certification for symmetric quadratic functions and\\
its connection to constrained Boolean hypercube\\optimization \thanks{Adam Kurpisz was supported by SNSF project PZ00P2$\_$174117, Aaron Potechin was supported in part by NSF grant CCF$\-$2008920.}}

\author{
  Adam, Kurpisz\\
  \small{ETH Z\"urich, Department of Mathematics, R\"amistrasse 101,	8092 Z\"urich, Switzerland}\\
  {adam.kurpisz@ifor.math.ethz.ch}\\
  \and
  Aaron, Potechin\\
  \small{University of Chicago, Department of Computer Science,  5730 S Ellis Ave, Chicago, IL 60637, US}\\
  {potechin@uchicago.edu} \\
  \and
  Elias Samuel, Wirth \\
  \small{TU Berlin, Institute of Mathematics, Strasse des 17. Juni 136, 10623 Berlin, Germany}\\
  {wirth@zib.de}
}

\date{}

\begin{document}
\maketitle
\begin{abstract}
We study the rank of the Sum of Squares (SoS) hierarchy over the Boolean hypercube for Symmetric Quadratic Functions (SQFs) in $n$ variables with roots placed in points $k-1$ and $k$. Functions of this type have played a central role in deepening the understanding of the performance of the SoS method for various unconstrained Boolean hypercube optimization problems, including the Max Cut problem. 
Recently, Lee, Prakash, de Wolf, and Yuen proved a lower bound on the SoS rank for SQFs of $\Omega(\sqrt{k(n-k)})$ and conjectured the lower bound of $\Omega(n)$ by similarity to a polynomial representation of the $n$-bit OR function.

Using Chebyshev polynomials, we refute the Lee---Prakash---de~Wolf---Yuen conjecture and prove that the SoS rank for SQFs is at most $O(\sqrt{nk}\log(n))$.

We connect this result to two constrained Boolean hypercube optimization problems. First, we provide a degree $O( \sqrt{n})$ SoS certificate that matches the known SoS rank lower bound for an instance of Min Knapsack, a problem that was intensively studied in the literature. Second, we study an instance of the Set Cover problem for which Bienstock and Zuckerberg conjectured an SoS rank lower bound of $n/4$. We refute the Bienstock---Zuckerberg conjecture and provide a degree $O(\sqrt{n}\log(n))$ SoS certificate for this problem.
\end{abstract}
\section{Introduction}
Semialgebraic proof systems, also called certificates of nonnegativity, are systematic methods to prove nonnegativity of polynomials over semialgebraic sets. One of the most successful approaches for constructing theoretically efficient algorithms for polynomial optimization problems is the Sum of Squares (SoS) certificate~\cite{GrigorievV01,Nesterov00,parrilo00,schor87}, 
For a wide variety of combinatorial optimization problems, SoS provides the best available algorithms~\cite{AroraRV09,GoemansW95,BarakRS11,GuruswamiS11,Lovasz79}.
The strength of this method has also come to light for Max CSP~\cite{LeeRagSteu15} and problems in robust estimation~\cite{KothariSS18}, dictionary learning~\cite{BarakKS15,SchrammS17}, tensor completion and decomposition~\cite{BarakM16,HopkinsSSS16,PotechinS17}, and problems arising from statistical physics~\cite{GhoshJJPR20}. 

However, the SoS algorithm also admits certain weaknesses. It is known to struggle with solving certain combinatorial optimization problems, e.g.,~\cite{BhaskaraCVGZ12,Cheung07,Grigoriev01b,KurpiszLM17,Tulsiani09}.
In a seminal example, Grigoriev showed that a $\Omega(n)$ degree SoS certificate is needed to detect a simple integrality argument for the {Knapsack} problem \cite{Grigoriev01}, see also~\cite{GrigorievHP02,KurpiszLM16,Laurent03a}.
A degree $n^{\Omega(\varepsilon)}$ SoS algorithm was proved to be unable to asymptotically certify an upper bound smaller than 2 times the optimal value for Sherrington-Kirkpatric Hamiltonian~\cite{BandeiraK19,GhoshJJPR20}.
Moreover, the degree $\Omega(\sqrt{n})$ SoS hierarchy was proved to have problems scheduling unit size jobs on a single machine to minimize the number of late jobs, see~\cite{KurpiszLM17b}, even though the problem is known to be solvable in polynomial time using the Moore-Hodgson algorithm~\cite{Moore68}. Finally, various examples where the SoS hierarchy fares very badly have been shown for the planted clique~\cite{BarakHKKMP16,MekaPW15} and Max CSP problems~\cite{KothariMOW17,ThapperZ17}.

The discrepancy between the excellent performance of the SoS hierarchy and its limitations has been studied extensively throughout the last decade. Thus, a natural question arises: what factors determine the difficulty of solving a problem for the SoS method?

A prominent example that was studied through the lens of this question is the Max Cut problem, which not only lies at the center of SoS research but was also one of the first problems for which lower bounds of the SoS rank were studied. Grigoriev proved that SoS needs at least degree $\lfloor\frac{n}{2} \rfloor$ to certify the size of the maximum cut in an odd clique of $n$ vertices~\cite{Grigoriev01}, for alternative proofs see also~\cite{GrigorievHP02,KurpiszLM16,Laurent03a}. In a breakthrough paper nearly two decades later, Parrilo showed that the Grigoriev's lower bound is tight by proving that every $n$-variate polynomial of degree $2$, nonnegative over the Boolean hypercube has an SoS certificate of degree at most $\lceil\frac{n}{2}\rceil$, see~\cite{FawziSP15}.
Subsequently, the analog of the results by Grigoriev and Parrilo for higher degree symmetric functions recently appeared in~\cite{KurpiszLM16c, SakueTKI17}, respectively.
Many of the problem instances with large lower bounds of the SoS rank target known limitations of the SoS method such as an issue with dealing with integrality constraints. Indeed, certifying the size of the maximum cut in a clique can be transformed into the problem of proving nonnegativity of the \emph{Symmetric Quadratic Function} (SQF) of the form $q_{\left\lceil\frac{n}{2}\right\rceil}(\mathbf{x})$ over the Boolean hypercube, where, throughout this paper, $q_k:\{0,1\}^n\to\R$ is a multivariate polynomial of the form
\begin{equation}\label{eq:SQF}
	q_k(\mathbf{x}):=(|\mathbf{x}|-k)(|\mathbf{x}|-k+1).
\end{equation}
The optimization of degree $2$ polynomials over the Boolean hypercube plays a central role in Theoretical Computer Science.
This claim is supported by the fact that high degree optimization problems attracted limited attention, especially since solving an NP-complete problem can be reduced in polynomial time to proving nonnegativity of a degree-$4$ even form~\cite{MurtyK87}. Moreover, if an SQF has a complex root with a corresponding conjugate root, the polynomial is globally nonnegative and admits an SoS certificate of degree $2$. Similarly, there exists an SoS certificate of nonnegativity of degree $2$ for SQFs over the Boolean hypercube if the roots are real and placed outside the interval $[0,n]$. Hence, the only interesting case is when the roots are real and located within some interval $[k-1,k]$ for $k \in \{1,\ldots,n\}$. 
Finding an SoS representation of the symmetric function $q_k$ has gained significant attention in the SoS community. However, up to this day, the exact SoS rank for $q_k$ is not known. The most recent result towards a characterization of the SoS rank of $q_k$ provides a lower and upper bound of the SoS degree that approximates the function $q_k$ with SoS polynomials in $l_1$ and $l_\infty$ norm~\cite{LeePWY16}. 
However, since finding an exact SoS certificate is at least as difficult as providing an approximate SoS representation, the result implies that for $k\geq 2$, $q_k$ does not admit an SoS certificate of degree smaller than $\Omega\left( \sqrt{k(n-k)} \right)$. Moreover, in~\cite{LeePWY16}, Lee, Prakash, de Wolf, and Yuen conjectured that the lower bound of the SoS approximate representation with error at most $\varepsilon$ in the $l_\infty$ norm is expected to be $ \Omega\left( \sqrt{k(n-k)}+ \sqrt{n \log(1/\varepsilon)} \right)$. They support the conjecture by arguing about similarity with approximating $n$-bit OR functions~\cite{Paturi92,Wolf10}. This conjecture, if true, would imply a lower bound on the exact SoS certificate for SQFs of $\Omega(n)$, even for small, constant values of $k$. Proving this conjecture is left as an open question in~\cite{LeePWY16}.
In this paper, we refute the Lee---Prakash---de Wolf---Yuen (LPdWY) conjecture. We show that certifying SQFs is easier than representing $n$-bit OR functions. More specifically, we prove the following theorem.
\begin{theorem}
	\label{thm:SQF_SoS_degree}
	For any $k \in \{2,\ldots, \lceil \frac{n}{2} \rceil \} $, there exists a degree $O(\sqrt{nk}\log (n)  )$ SoS certificate of nonnegativity for the Boolean function $q_k$ as in~\eqref{eq:SQF}. 
\end{theorem}

We motivate the research on the SoS degree of the SQFs $q_k$ by connecting it to two combinatorial optimization problems.
We first consider the instance of the \textsc{Min Knapsack} (MK) problem. For $P\geq 2$, the problem is defined as: 
\begin{eqnarray}
	\mbox{MK:}
	\qquad \min  \sum_{i \in [n]} x_i   & \quad \mbox{s.t.} &  \sum_{i \in [n] } x_i \geq \frac{1}{P},   \qquad \mathbf{x}\in\{0,1\}^n .
	\label{eq:MK_def}
\end{eqnarray}
For $P=2$, the problem was previously considered by Cook and Dash \cite{cook2001matrix}. They proved that the Lovasz-Schrijver hierarchy rank is $n$. For the Sherali-Adams hierarchy, Laurent proved that the rank is also equal to $n$ and raised the open question to find the rank for the SoS hierarchy~\cite{Laurent03}.
For $n = 2$, they also proved that the SoS rank is $2$, but the discussion for general $n$ was left as an open question.
Currently, it is known that the SoS rank of the MK problem for $P=2$ falls within $\Omega(\sqrt{n})$ and $\lceil \frac{ n+ 4\lceil \sqrt{n} \rceil   }{2} \rceil$, see~\cite{Kurpisz19}. In this paper, we prove an upper bound on the SoS rank for the MK problem.
\begin{theorem}\label{thm:MK_SoS_rank}
	The SoS rank for the MK problem is $\Omega(\sqrt{n}\log(P))$.
\end{theorem}
The existing lower bound for general $P$ (see Lemma 14 of \cite{Kurpisz19}) is $\Omega(\sqrt{n\log(P)})$, so this is tight when $P$ is constant, though for larger $P$ there is a gap of $O(\sqrt{\log(P)})$.

We also consider the following instance of the \textsc{Set Cover} (SC) problem:
\begin{eqnarray}
	\mbox{SC:}
	\qquad \min  \sum_{i \in [n]} x_i   & \quad \mbox{s.t.} &  \sum_{i \in [n] \setminus\{j\}} x_i \geq 1 \qquad \forall j \in [n], \qquad \mathbf{x}\in\{0,1\}^n.
\end{eqnarray}
This instance was considered in~\cite{BienstockZ04} and it is known that the SoS hierarchy cannot solve this problem with a degree smaller than $\Omega(\sqrt{n})$~\cite{Kurpisz19}.
In~\cite{BienstockZ04}, Bienstock and Zuckerberg raised the question of what the actual SoS rank of this polytope is, conjecturing that, based on numerical experiments, the SoS rank is at least $\frac{n}{4}$. 
In this paper, using the SoS certificate for SQFs in Theorem~\ref{thm:SQF_SoS_degree}, we refute the Bienstock---Zuckerberg conjecture and provide a nearly tight SoS rank for the SC problem:%
\begin{theorem}\label{thm:SC_SoS_rank}
	The SoS rank for the SC problem is at most $O(\sqrt{n} \log (n))$.
\end{theorem}

\section{Preliminaries}\label{sec:Preliminaries}

For $n\in \mathbb{N}$, let $[n]=\{1,\ldots,n\}$. For $\mathbf{x}\in \mathbb{R}^n$, let $\mathbb{R}[\mathbf{x}]=\mathbb{R}[x_1,\ldots,x_n]$ be the ring of \emph{$n$-variate real polynomials}.
For a set of polynomials $\cG \subseteq \R[x]$, the corresponding \emph{semialgebraic set} is
\begin{align*}
	\cG_+ \ := \ \{ \mathbf{x} \in \mathbb{R}^n~ |~ g(\mathbf{x})\geq 0 \text{ for all } {g \in \cG}  \} \subseteq \R^n.
\end{align*}
Throughout this paper, we consider optimization problems on the Boolean hypercube $\{0,1\}^n$ and therefore, for $\cH:=\{ \pm (x_1^2-x_1),\ldots,\pm(x_n^2-x_n) \}$, we assume that $\cG$ is of the form
\begin{align*}
	\cG \ := \ \cH \cup \{g_1,\ldots,g_m:~g_i \in \mathbb{R}[\mathbf{x}] \text{ for all } ~ i \in [m] \} ,
\end{align*}
where $m\in \N_{> 0}$. This implies that $\mathcal{G}_+\subseteq \{0,1\}^n$.
Moreover, define the \emph{cone of nonnegative polynomials with respect to a given semialgebraic set, $\mathcal{G}_+$,} as
\begin{align*}
	\cK(\cG_+) \ := \ \{f \in \R[\mathbf{x}] \ | \ f(\mathbf{x}) \geq 0 \text{ for all } \mathbf{x} \in \cG_+\}.
\end{align*}
For given $f \in \R[\mathbf{x}]$ and $\cG\subseteq \R[\mathbf{x}]$, define the corresponding \emph{Constrained Polynomial Optimization Problem} (CPOP) as
\begin{align*}
	\label{eq:intro_POP}
	\begin{aligned}
		f^* \ := \ \min\{f(\mathbf{x})~ |~ \mathbf{x} \in \cG_+\} \ = \ \max\{\lambda \in \mathbb{R}~ |~ f-\lambda \in \cK(\cG_+)\}.
	\end{aligned}
\end{align*}
Generally, since CPOP is NP-hard, it is desirable to find a proper subset that is a good inner approximation of $\cK(\cG_+) $ such that the corresponding program is computationally \emph{tractable}.
The \emph{SoS method} approximates the cone $\cK(\cG_+)$ by using the set of \emph{sum of square polynomials}. We define the set of finite sum of squares polynomials as
$\Sigma := \{s\mid s = \sum_{i=1}^{k}s_i^2, s_i\in \mathbb{R}[\mathbf{x}] \  \forall i\in [k], k \in \mathbb{N}_{>0}\}$
and let 
$\Sigma_{n,d}:= \{s\mid s = \sum_{i=1}^{k}s_i^2, s_i\in \mathbb{R}[\mathbf{x}] \text{ and } \deg(s_i) \leq d \ \forall i\in [k], k \in \mathbb{N}_{>0}\}$
denote the
polynomials which are sums of squares of polynomials of degree at most $d$.
We define the \emph{hierarchy of certificates of nonnegativity depending on $d,n \in\mathbb{N}$} as
\begin{equation*}
	\Sigma_{n,d}^{\cG}:= \left\{ s_0+\sum_{i=1}^m s_i g_i~|~s_i \in\Sigma_{n,d},~g_i \in \cG \ \forall i \in [m] \ \text{and} \ s_0\in \Sigma_{n, 2 \left\lceil \frac{2d +\deg(\cG)}{2} \right\rceil} \right\},
\end{equation*}
where $\deg(\cG)=\max\{\deg(g)~|~g \in \cG\}$.
The \emph{degree $d$ SoS certificate} for $f$ being nonnegative over $\cG_+$ is $f \in \Sigma_{n,d}^{\cG}$.  Moreover, throughout the paper we say that a multivariate polynomial $f$ is \emph{a degree $d$ SoS modulo Boolean axioms} if $f \in \Sigma_{n,d}^\cH$.
The \emph{degree $d$ SoS program} for CPOP~is
\begin{equation}
	\label{eq:intro_SoS_d_CPOP}
	f^{d}_\Sigma \ := \ \quad \max\{\lambda \in \mathbb{R}~ |~ f-\lambda \in \Sigma_{n,d}^{\cG}\}
\end{equation}
and is called \emph{exact} if $f^d_\Sigma=f^*$. The smallest degree $d$ such that the degree $d$ SoS program is exact is called the \emph{SoS rank}. 
Over the Boolean hypercube, the degree $d$ SoS program can be solved via a \emph{semidefinite program} (SDP) 
of size $O(m \sum_{k=0}^d \binom{n}{k})$. Moreover, the degree $n$ SoS program is exact, see, e.g.,~\cite{BarakS16,Lasserre01z,Laurent03}.

Throughout this paper, we often encounter the following type of multivariate polynomials.
\begin{definition}\label{def: symmetric polynomial}
	A polynomial $f:~\{0,1\}^n\to\R$ is \emph{symmetric} if there exists a univariate polynomial $\tilde{f}~:\R\to\R$ such that
	$$
	f(\mathbf{x})=\tilde{f}\left(\sum_{i=1}^n x_i\right)
	$$
	for all $\mathbf{x}\in \{0,1\}^n$.
\end{definition}
With this in mind, let $|\mathbf{x}|:= \sum_{i=1}^nx_i$ for any $\mathbf{x}\in \{0,1\}^n$.
To prove SoS rank upper bounds, we consider symmetric multivariate polynomials over $\{0,1\}^n$ as univariate polynomials over $[0,n]$ and apply one of the many results on SoS certificates for univariate polynomials. 
\begin{remark}
	\label{rem:univariate_to_multivaraite_SoS}
	Throughout this paper, we make frequent use of the fact that SoS certificates for polynomials over $[0,n]$ translate to SoS certificates for symmetric polynomials over $\{0,1\}^n$. More formally, if a univariate polynomial $\tilde{f}:\R\to \R$ has an univariate SoS certificate of degree $d$ on $[0,n]$, then the multivariate polynomial $f:\{0,1\}^n\to\R$ such that $f(\mathbf{x}):=\tilde{f}(|\mathbf{x}|)$ has a degree $d$ SoS certificate of nonnegativity over the Boolean hypercube.
\end{remark}
In this paper, we use the following theorem to prove the SoS rank for univariate polynomials.

\begin{theorem}[{\cite[Theorem 3.72]{blekherman2012semidefinite}}]\label{thm:Blekherman_TheSecond}
	Let $a<b$. Then the univariate polynomial $p(x)$ is nonnegative on $[a,b]$ if and only if it can be written as
	$$
	\begin{cases}
		p(x)=s(x)+(x-a)(b-x)\cdot t(x) & \text{if} \ \deg(p) \ \text{is even,}\\
		p(x)=(x-a) \cdot s(x)+(b-x)\cdot t(x) & \text{if} \ \deg(p) \ \text{is odd,}\\
	\end{cases}
	$$
	where $s,t$ are sum of squares. In the first case, we have $\deg(p)=2d$, $\deg(s)\leq 2d$, and $\deg(t)\leq 2d-2$. In the second, $\deg(p)=2d+1$, $\deg(s)\leq 2d$, and $\deg(t)\leq 2d-2$.
\end{theorem}

Finally, throughout the paper we use degree-$d$ Chebyshev polynomials of the first type, which were used in several applications for bounds of sum of squares ranks, i.e., \cite{Kurpisz19, SlotL19, potechin2020sum}.
We frequently use the following lemma.
\begin{lemma}
	\label{lem:Chebyshev_1-c/n_properties}
	Let $n , d \in \mathbb{N}$ such that $d \leq n$. Then,
	\begin{enumerate}
	    \item For all $c\in [0, n]$,
	$$
	T^2_{ d}\left(-1-\frac{c}{n}\right) \geq \frac{1}{4}\left(-1-\sqrt{\frac{2c}{n}}\right)^{2d }$$
	and
	$$
	T^2_{ d}\left(-1-\frac{c}{n}\right) \leq \left(-1-2\sqrt{\frac{2c}{n}}\right)^{2d }.
	$$
	Moreover, for constant $c$ and $n$ big enough, 
	$$
	T^2_{ d}\left(-1-\frac{c}{n}\right) \leq \left(-1-\sqrt{\frac{2c+1}{n}}\right)^{2d }.
	$$
	\item For all $c\in (n, \infty)$, $T^2_{ d}\left(-1-\frac{c}{n}\right) \leq  \left(-1-3\frac{c}{n}\right)^{2d }$.
	\end{enumerate}
\end{lemma}
\begin{proof}
    It holds that:
	\begin{enumerate}
	    \item Consider the characterization of Chebyshev polynomials for $x \geq |1|$ given in \cite[Equation 1.12]{Rivlin74}:
$$
		T_d(x)=\frac{1}{2}\left( \left( x-\sqrt{x^2-1}\right)^d + \left( x+\sqrt{x^2-1}\right)^d\right).
$$
	For $x =-1-\frac{c}{n}$ and $c  \in [0, n]$, we have
	$$
	T^2_{ d}\left(-1-\frac{c}{n}\right) \geq \frac{1}{4} \left( \left(-1-\frac{c}{n}\right)-\sqrt{\left(-1-\frac{c}{n}\right)^2-1}\right)^{ 2d} 
	\geq \frac{1}{4}\left(-1-\sqrt{\frac{2c}{n}}\right)^{2d }
	$$
	and
	\begin{align}
		T^2_{ d}\left(-1-\frac{c}{n}\right) \leq & \left( \left(-1-\frac{c}{n}\right)-\sqrt{\left(-1-\frac{c}{n}\right)^2-1}\right)^{ 2d} \nonumber \\
		& \leq  \left( \left(-1-\frac{c}{n}\right)-\sqrt{\frac{2c}{n} +\frac{c^2}{n^2}  }\right)^{ 2d} \nonumber \\
		& \leq  \left( -1-\sqrt{\frac{c}{n}}-\sqrt{\frac{2c}{n} +\frac{c}{n}  }\right)^{ 2d} \leq  \left(-1-2\sqrt{\frac{2c}{n}}\right)^{2d }.
	\end{align}
	Moreover, we have
	$$
	T^2_{ d}\left(-1-\frac{c}{n}\right) \leq  \left( \left(-1-\frac{c}{n}\right)-\sqrt{\left(-1-\frac{c}{n}\right)^2-1}\right)^{ 2d} 
	\leq \left(-1-\sqrt{\frac{2c+1}{n}}\right)^{2d },
	$$
	where the last inequality holds for $n$ large compared to $c$.
	\item For $x =-1-\frac{c}{n}$ and $c  \in (n, \infty)$, we have
	\begin{align}
		T^2_{ d}\left(-1-\frac{c}{n}\right) \leq & \left( \left(-1-\frac{c}{n}\right)-\sqrt{\left(-1-\frac{c}{n}\right)^2-1}\right)^{ 2d} \nonumber \\
		& \leq  \left( \left(-1-\frac{c}{n}\right)-\sqrt{\frac{2c}{n} +\frac{c^2}{n^2}  }\right)^{ 2d} \nonumber \\
		& \leq  \left( -1-\frac{c}{n}-\sqrt{\frac{2c^2}{n^2} +\frac{c^2}{n^2}  }\right)^{ 2d} \leq  \left(-1-3\frac{c}{n}\right)^{2d }.
	\end{align}
	\end{enumerate}
\end{proof}
\section{SoS rank for SQFs}\label{section:SoSRankForSQFs}
In this section, we refute the LPdWY conjecture stated in~\cite{LeePWY16} by proving Theorem~\ref{thm:SQF_SoS_degree}.
To prove Theorem~\ref{thm:SQF_SoS_degree}, it is sufficient to prove the following theorem. 
\begin{theorem}
\label{thm:exists_s(x)_st_q(x)-s(x)geq0}
	For all $n \in \mathbb{N}$ and all $k \in [n]$, there exists a univariate polynomial $s(x)$ of degree $O(\sqrt{kn}\log(n))$ such that 
	\begin{enumerate}
		\item $s\left(\sum_{i=1}^{n}{x_i}\right)$ is a sum of squares (modulo the Boolean axioms).
		\item For all $x \in [0,n]$, $(x-k+1)(x-k) - s(x) \geq 0$.
	\end{enumerate}
\end{theorem}

Indeed, by Theorem~\ref{thm:exists_s(x)_st_q(x)-s(x)geq0} and Theorem 6, there exist sum of squares polynomials $s,~s_1$ and $s_2$ of degree $O(\sqrt{kn} \log(n))$ s.t.
$$
(x-k+1)(x-k) = s(x) + s_1(x) +s_2(x)x(n-x).
$$
We now make the following observations:
\begin{enumerate}
    \item By Theorem 8, $s(\sum_{i=1}^n x_i)$ is a sum of squares polynomial modulo the Boolean axioms.
    \item $s_1(\sum_{i=1}^n x_i),~s_2(\sum_{i=1}^n x_i)$ are sum of squares polynomials.
    \item $\sum_{i=1}^n x_i  = \sum_{i=1}^n x_i^2 - \sum_{i=1}^n\left(x_i^2- x_i\right)$ is a sum of squares polynomial modulo the Boolean axioms.
    \item $n-\sum_{i=1}^n x_i = \sum_{i=1}^n \left(1-x_i  \right) = \sum_{i=1}^n \left( \left( x_i -1  \right)^2 - \left(x_i^2 -x_i \right) \right)$ is a sum of squares polynomial modulo the Boolean axioms.
\end{enumerate}
Putting everything together, the multivariate polynomial $q_k(\mathbf{x})$ has an $O(\sqrt{kn} \log(n))$ SoS certificate modulo the Boolean axioms of the form
$$
q_k\left(\mathbf{x}\right) =  s\left(\sum_{i=1}^n x_i\right) + s_1\left(\sum_{i=1}^n x_i\right) +s_2\left(\sum_{i=1}^n x_i\right)\left( \sum_{i=1}^n x_i \right) \left(n-\sum_{i=1}^n x_i\right).
$$
Before we prove Theorem~\ref{thm:exists_s(x)_st_q(x)-s(x)geq0}, we make the following observation which shows that our upper bound for $q_k(x)$ applies for any symmetric quadratic function with roots in $[k-1,k]$.
\begin{corollary}
	\label{cor:SQF_SoS_degree_closer_roots_crude_bound}
	For any $k \in \{1, \ldots , \lceil n/2 \rceil\}$ and any $a\leq b \in [k-1,k]$, a polynomial
	$$
		f_k:=(x-a)(x-b)
	$$
	admits an SoS certificate over the Boolean hypercube of degree at most the degree of an SoS certificate over the Boolean hypercube for polynomial $q_k$.
\end{corollary}
\begin{proof}
	We have $f_k(x) \geq \left((k-a)(b-k+1) + (k-b)(a-k+1)\right)q_k(x)$ as
	{\small
    \begin{align*}
    &(|x| - a)(|x| - b) \\
    &= \left((k-a)(|x|- k + 1) + (a-k+1)(|x| - k)\right)\left((k-b)(|x|- k + 1) + (b-k+1)(|x| - k)\right) \\
    &= (k-a)(k-b)(|x|- k + 1)^2 + (a-k+1)(b-k+1)(|x|- k)^2 \\
    &+\left((k-a)(b-k+1) + (k-b)(a-k+1)\right)(|x| - k + 1)(x - |k|)
    \end{align*}	}
	and invoke Theorem~\ref{thm:SQF_SoS_degree} to conclude the proof.
\end{proof}

\subsection{Proof of Theorem~\ref{thm:exists_s(x)_st_q(x)-s(x)geq0}}
We construct $s(x)$ in two steps. We first construct a polynomial $s_1(x)$ which is a sum of squares (modulo the Boolean axioms), is less than or equal to $(x-k+1)(x-k)$ on the interval $[0,2k-1]$, and is not too large on the interval $[2k-1,n]$. We then construct a polynomial $s_2(x)$ which is a sum of squares, is less than or equal to $1$ on the intervals $[0,k-1]$ and $[k,2k-1]$, is greater than or equal to $1$ on the interval $[k-1,k]$, and is very small on the interval $[2k-1,n]$. We then take $s(x) = s_1(x)s_2(x)$.
More precisely, we have the following conditions on $s_1$ and $s_2$:
\begin{enumerate}
	\item $s_1\left(\sum_{i=1}^{n}{x_i}\right)$ is a sum of squares (modulo the Boolean axioms) and $s_2(x)$ is a sum of squares.
	\item For all $x \in [k-1,k]$, $\frac{s_1(x)}{(x-k+1)(x-k)} \geq 1$ and $s_2(x) \geq 1$.
	\item For all $x \in [0,k-1] \cup [k,2k-1]$, $\frac{s_1(x)}{(x-k+1)(x-k)} \leq 1$ and $s_2(x) \leq 1$.
	\item For all $x \in [2k-1,n]$, $\left|\frac{s_1(x)}{(x-k+1)(x-k)}\right| \leq n^{40k}$ and $s_2(x) \leq n^{-40k}$.
	\item $s_1(x)$ has degree $O(k)$ and $s_2(x)$ has degree $O(\sqrt{nk}\log(n))$.
\end{enumerate}
\begin{proposition}
	If $s_1(x)$ and $s_2(x)$ satisfy the above conditions and we take $s(x) = s_1(x)s_2(x)$ then $s\left(\sum_{i=1}^{n}{x_i}\right)$ is a sum of squares (modulo the Boolean axioms) and for all $x \in [0,n]$, $(x-k+1)(x-k) - s(x) \geq 0$.
\end{proposition}
\begin{proof}
	We make the following observations:
	\begin{enumerate}
		\item Since $s_1\left(\sum_{i=1}^{n}{x_i}\right)$ is a sum of squares (modulo the Boolean axioms) and $s_2(x)$ is a sum of squares, the product $s\left(\sum_{i=1}^{n}{x_i}\right) = s_1\left(\sum_{i=1}^{n}{x_i}\right)s_2\left(\sum_{i=1}^{n}{x_i}\right)$ is a sum of squares (modulo the Boolean axioms).
		\item For all $x \in [0,k-1] \cup [k,2k-1]$, since $(x-k+1)(x-k) \geq 0$, $\frac{s_1(x)}{(x-k+1)(x-k)} \leq 1$, and $0 \leq s_2(x) \leq 1$, 
		\[
		(x-k+1)(x-k) - s(x) = (x-k+1)(x-k)\left(1 - s_2(x)\frac{s_1(x)}{(x-k+1)(x-k)}\right) \geq 0.
		\]
		\item For all $x \in [k-1,k]$, since $(x-k+1)(x-k) \leq 0$, $\frac{s_1(x)}{(x-k+1)(x-k)} \geq 1$,  and $s_2(x) \geq 1$, 
		\[
		(x-k+1)(x-k) - s(x) = (x-k+1)(x-k)\left(1 - s_2(x)\frac{s_1(x)}{(x-k+1)(x-k)}\right) \geq 0.
		\]
		\item For all $x \in [2k-1,n]$, since $(x-k+1)(x-k) \geq 0$, $\left|\frac{s_1(x)}{(x-k+1)(x-k)}\right| \leq n^{40k}$ and $|s_2(x)| \leq n^{-40k}$,
		\[
		(x-k+1)(x-k) - s(x) = (x-k+1)(x-k)\left(1 - s_2(x)\frac{s_1(x)}{(x-k+1)(x-k)}\right) \geq 0.
		\]
	\end{enumerate}
\end{proof}
Thus, we have an SoS proof of degree $O(\sqrt{kn}log(n))$ that $(|x|-k+1)(|x|-k) \geq 0$.
\subsubsection{Constructing the polynomial $s_1(x)$}
We now construct the polynomial $s_1(x)$.
\begin{lemma}
	For $n\in \N$ and all $k \in [n]$, there exists a polynomial $s_1(x)$ such that
	\begin{enumerate}
		\item $s_1\left(\sum_{i=1}^{n}{x_i}\right)$ has a degree $O(k)$ sum of squares (modulo the Boolean axioms) certificate.
		\item For all $x \in [k-1,k]$, $\frac{s_1(x)}{(x-k+1)(x-k)} \geq 1$.
		\item For all $x \in [0,k-1] \cup [k,2k-1]$, $\frac{s_1(x)}{(x-k+1)(x-k)} \leq 1$.
		\item For all $x \in [2k-1,n]$, $\left|\frac{s_1(x)}{(x-k+1)(x-k)}\right| \leq n^{40k}$.
	\end{enumerate}
\end{lemma}
\begin{proof}
	For $k = 1$, we can take $s_1(x) = x(x-1)$ so we can assume that $n \geq k \geq 2$. For $k \geq 2$, we use the following construction.\footnote{Definitions~\ref{def:g_k} and~\ref{def:s_k} are only used in the current section, Section~\ref{section:SoSRankForSQFs}.}
	\begin{definition}\label{def:g_k}
		For all natural numbers $k \geq 2$, define $g_k(x)$ to be the polynomial 
		\[
		g_k(x) = x^{16k}(x - 2k + 1)^{16k}\prod_{i \in \{0,\ldots,2k-1\} \setminus \{k-1,k\}} (x-i).
		\]
	\end{definition}
	\begin{definition}\label{def:s_k}
		Given a natural number $n$ and $k \in \{2,3,\ldots,n\}$, we define $s_1(x)$ as follows:
		\begin{enumerate}
			\item If $k$ is odd, then we define $s_1(x) = \frac{g_k(x)}{g_k(k-1)}(x-k+1)(x-k)$.
			\item If $k$ is even, then we define $s_1(x) = -\frac{g_k(x)(x+1)(x-2k)}{g_k(k-1)k(k+1)}(x-k+1)(x-k)$.
		\end{enumerate}
	\end{definition}
	We verify the desired properties. We first show that $s_1\left(\sum_{i=1}^{n}{x_i}\right)$ is a sum of squares (modulo the Boolean axioms). If $k$ is odd, then since $g_k(k-1) > 0$, $\prod_{i=0}^{2k-1} \left(\left(\sum_{i=1}^{n}{x_i}\right)-i\right)$  is a sum of squares (modulo the Boolean axioms), and by~\cite[Lemma 4.4]{LeePWY16}, \[
	s_1\left(\sum_{i=1}^{n}{x_i}\right) = \frac{\left(\sum_{i=1}^{n}{x_i}\right)^{16k}(\left(\sum_{i=1}^{n}{x_i}\right) - 2k + 1)^{16k}}{g_k(k-1)}\prod_{i=0}^{2k-1} \left(\left(\sum_{i=1}^{n}{x_i}\right)-i\right)
	\]
	is a sum of squares (modulo the Boolean axioms). If $k$ is even, then since $g_k(k-1) < 0$, $\left(\sum_{i=1}^{n}{x_i}\right) + 1$ and $\prod_{i=0}^{2k} \left(\left(\sum_{i=1}^{n}{x_i}\right)-i\right)$ are sum of squares (modulo the Boolean axioms),
	\[
	s_1\left(\sum_{i=1}^{n}{x_i}\right) = -\frac{\left(\sum_{i=1}^{n}{x_i}\right)^{16k}(\left(\sum_{i=1}^{n}{x_i}\right) - 2k + 1)^{16k}}{g_k(k-1)k(k+1)}\left(\left(\sum_{i=1}^{n}{x_i}\right) + 1\right)\prod_{i=0}^{2k} \left(\left(\sum_{i=1}^{n}{x_i}\right)-i\right)
	\]
	is a sum of squares (modulo the Boolean axioms). Finally, to argue about the degree, note that by~\cite[Lemma 4.4]{LeePWY16}, $\prod_{i=0}^{2k-1} \left(\left(\sum_{i=1}^{n}{x_i}\right)-i\right)$ has a sum of squares (modulo the Boolean axioms) certificate of degree $2k$ and thus, for all $k$, $s_1\left(\sum_{i=1}^{n}{x_i}\right)$ has a sum of squares (modulo the Boolean axioms) certificate of degree $O(k)$.
	
	For the fourth property, observe that for $x \in [0,n]$, every term in the numerator (except for $(x+1)$ when $k$ is even) has magnitude at most $n$, every term in the denominator has magnitude at least $1$, and there are less than $40k$ terms in the numerator.
	
	The second and third properties follow immediately from the following lemma.
	\begin{lemma}
		For all natural numbers $k \geq 2$, $g_k(x)$ satisfies the following properties:
		\begin{enumerate}
			\item For all $x \in [0,2k-1]$, $g_k(2k-1-x) = g_k(x)$.
			\item For all $x \in [k-1,k]$, $\frac{g_k(x)}{g_k(k-1)} \geq 1$.
			\item For all $x \in [0,k-1] \cup [k,2k-1]$, $\left|\frac{g_k(x)}{g_k(k-1)}\right| \leq 1.$ 
		\end{enumerate}
	\end{lemma}
	\begin{proof}
		Since the first and second properties hold for every term in the product
		$$g_k(x) = (-1)^{k-1}\left(x(x - 2k + 1)\right)^{16k}\left(\prod_{i=0}^{k-2}{(x-i)(2k-1-x-i)}\right),$$
		they hold for $g_k(x)$ as well. 
		
		By symmetry, it suffices to show the third property for $x \in [0,k-1]$. For $x \in \{0,1,\ldots,k-2\}$, $g_k(x) = 0$ and for $x \in (k-2,k-1)$, the third property holds for every term in this product, so it holds for $g_k(x)$ as well. 
		To show that the third property holds for $x \in [0,k-2] \setminus \{0,1,\ldots,k-2\}$, we compare $g_k(x-m)$ and $g_k(x)$, where $x \in (k-2,k-1)$ and $m \in \{0,1,\ldots,k-2\}$. For this, we decompose $g_k(x)$ as $g_k(x) = a_k(x)b_k(x)^{16k}$, where $a_k(x) = \prod_{i \in \{0,\ldots,2k-1\} \setminus \{k-1,k\}} (x-i)$ and $b_k(x) = x(2k-1-x)$.
		\begin{lemma}\label{aklemma}
			Let $a_k(x) = \prod_{i \in \{0,\ldots,2k-1\} \setminus \{k-1,k\}} (x-i) =  \left(\prod_{i=0}^{k-2}{(x-i)}\right)\left(\prod_{i=k+1}^{2k-1}{(x-i)}\right)$. For all $x \in (k-2,k-1)$ and all $m \in \{1,\ldots,k-2\}$, 
			$\left|\frac{a_k(x-m)}{a_k(x)}\right| \leq e^{\frac{16m^2}{k}}$.
		\end{lemma}
		\begin{proof}
            Observe that 
			\begin{align*}
				\left|\frac{a_k(x-m)}{a_k(x)}\right| &= \left|\frac{\prod_{j=1}^{m}{(x-k+2-j)}}{\prod_{j=1}^{m}{(x+1-j)}} \cdot \frac{\prod_{j=1}^{m}{(x-2k+1-j)}}{\prod_{j=1}^{m}{(x-k-j)}}\right| \\
				&= \left|\frac{\prod_{j=1}^{m}{(k-2-x+j)}}{\prod_{j=1}^{m}{(k - x + j)}} \cdot \frac{\prod_{j=1}^{m}{(2k - x-1+j)}}{\prod_{j=1}^{m}{(x -m + j)}}\right| \\ 
				&\leq \left|\prod_{j=1}^{m}{\left(\frac{k+1+j}{k-2-m+j}\right)}\right|. 
			\end{align*}
			We distinguish between two cases.
			\begin{enumerate}
				\item If $m \leq \frac{3k}{4} - 1$, observe that 
				\begin{align*}
					\left|\prod_{j=1}^{m}{\left(\frac{k+1+j}{k-2-m+j}\right)}\right| &= \prod_{j=1}^{m}{\left(1 + \frac{m+3}{k-2-m+j}\right)} \\
					&\leq \prod_{j=1}^{m}{\left(1 + \frac{m+3}{k-m-1}\right)} \leq \prod_{j=1}^{m}{e^{\frac{m+3}{k-m-1}}} = e^{\frac{m(m+3)}{(k - m + 1)}} \leq e^{\frac{16m^2}{k}}.
				\end{align*}
				\item If $m > \frac{3k}{4} - 1$, then $m \geq \frac{3k}{4} - \frac{3}{4} \geq \frac{3k}{8}$ (as $k \geq 2$). Thus,
				\[
				\left|\prod_{j=1}^{m}{\left(\frac{k+1+j}{k-2-m+j}\right)}\right| \leq \prod_{j=1}^{k-2}{\left(\frac{k+1+j}{j}\right)} = \frac{(2k-1)!}{(k-2)!(k+1)!} \leq 2^{2k-1} \leq e^{\frac{16m^2}{k}}.
				\]
			\end{enumerate}
		\end{proof}
		\begin{lemma}\label{bklemma}
			Let $b_k(x) = x(2k-1-x)$. For $x \in (k-2,k-1)$ and $m \in [k-2]$,~$
			\left|\frac{b_k(x-m)}{b_k(x)}\right| \leq e^{-\frac{m^2}{k^2}}.
			$
		\end{lemma}
		\begin{proof} 
			Observe that 
			\begin{align*}
				\frac{b_k(x-m)}{b_k(x)} = \frac{(x-m)(2k-1+m-x)}{x(2k-1-x)} & = \frac{x(2k-1-x) - (2k-1-2x)m - m^2}{x(2k-1-x)} \\
				& \leq 1 - \frac{m^2}{x(2k-1-x)} \leq 1 - \frac{m^2}{k^2} \leq e^{-\frac{m^2}{k^2}}.
			\end{align*}
			\vspace{-0.6cm}
		\end{proof}  
		\begin{corollary}
			For all $x \in (k-2,k-1)$ and $m \in \{1,\ldots,k-2\}$, 	        $\left|\frac{g_k(x-m)}{g_k(x)}\right| \leq 1$.
		\end{corollary}
		\begin{proof} 
			By Lemmas \ref{aklemma} and \ref{bklemma},
			$$
			\left|\frac{g_k(x-m)}{g_k(x)}\right| = \left|\frac{a_k(x-m)}{a_k(x)}\right| \left|\frac{b_k(x-m)}{b_k(x)}\right|^{16k}
			\leq e^{\frac{16m^2}{k}} \left(e^{-\frac{m^2}{k^2}}\right)^{16k}~=~1.
			$$
		\end{proof} 
		\end{proof}
\end{proof}
\subsubsection{Constructing the polynomial $s_2(x)$}
We now construct the polynomial $s_2(x)$.
\begin{lemma}\label{lem:h_sufficient_cond}
	For all $n\in \N$ and all $k \in [n]$, there exists a polynomial $s_2(x)$ of degree $O(\sqrt{kn}log(n))$ satisfying the following properties:
	\begin{enumerate}
		\item $s_2(x)$ is a sum of squares.
		\item For all $x \in [k-1,k]$, $s_2(x) \geq 1$.
		\item For all $x \in [0,k-1] \cup [k,2k-1]$, $s_2(x) \leq 1$.
		\item For all $x \in [2k-1,n]$, $s_2(x) \leq n^{-40k}$.
	\end{enumerate}
\end{lemma}
\begin{proof}
	\begin{lemma}
		\label{lem:H_polynpomial_properties}
		For $C:=  e^{8\sqrt{3}}$ and $k \in \{0, \ldots, \lceil n/2 \rceil \}$, $H_k(x) :=  T^2_{\sqrt{\frac{n}{k}}}\left( 2\frac{x}{n} -1 -2\frac{2k-1}{n}   \right)$ satisfies the following properties:
		\begin{enumerate}
			\item For all $x \in [2k-1,n]$, $H_k(x) \leq 1$.
			\item For all $k \in [0,2k-1]$, $H'_k(x) < 0$. \label{prop: 2}
			\item $H_k(0) \leq C$. \label{prop: 3}
			\item $H_k(k) \geq 1.5$.
		\end{enumerate}
	\end{lemma}
	\begin{proof}
		Note that
		$H_k(2k-1)=T^2_{\sqrt{\frac{n}{k}}}\left(-1 \right) =1$ and $H_k(n)=T^2_{\sqrt{\frac{n}{k}}}\left(1 - 2\frac{2k-1}{n}  \right) \leq 1$, which implies the first property.
		We prove Properties~\eqref{prop: 2} and~\eqref{prop: 3}. By Lemma~\ref{lem:Chebyshev_1-c/n_properties}, for $k$ such that $4k-2 \leq n$, we have
		$$
		H_k(0)=T^2_{\sqrt{\frac{n}{k}}}\left(-1 -\frac{4k-2}{n}  \right) \leq 
		  \left(1 +  \sqrt{\frac{32k-16}{n}}  \right)^{2\sqrt{\frac{n}{k}}}
		  \leq  e^{2\sqrt{\frac{32k-16}{k}}} \leq  e^{8\sqrt{3}}$$
		and for $k$ such that $4k-2 \geq n$, by Lemma~\ref{lem:Chebyshev_1-c/n_properties}, for $c \geq n$, we have
		$$
		H_k(0)=T^2_{\sqrt{\frac{n}{k}}}\left(-1 -\frac{4k-2}{n}  \right) \leq 
		  \left(1 +  \frac{12k}{n}  \right)^{2\sqrt{\frac{n}{k}}}
		  \leq \left(1 + \sqrt{\frac{12k}{n}}  \right)^{4\sqrt{\frac{n}{k}}}
		\leq  e^{4\sqrt{\frac{12k}{k}}} \leq  e^{8\sqrt{3}}.$$
		Moreover, by Lemma~\ref{lem:Chebyshev_1-c/n_properties} we have 
		\[
		H_k(k)=T^2_{\sqrt{\frac{n}{k}}}\left(  -1 -\frac{2k-2}{n}  \right) \geq \frac{1}{4} \left(1+\sqrt{\frac{4k-4}{n}}  \right)^{2\sqrt{\frac{n}{k}}} \geq \frac{1}{4} \left(1+\sqrt{\frac{2k}{n}}  \right)^{2\sqrt{\frac{n}{k}}},
		\]
		where the last inequality holds because $k \geq 2$.  Finally, since $n \geq 2k$, 
		\[
		\frac{1}{4} \left(1+\sqrt{\frac{2k}{n}} \right)^{2\sqrt{\frac{n}{k}}} \geq \frac{1}{4}2^{2  \sqrt{\frac{n}{k}}  \sqrt{\frac{2k}{n}}} = \frac{1}{4}2^{2\sqrt{2}} \geq 1.5.
		\]

	\end{proof}
	
	\begin{lemma}
	    For any constants $a,b,C$ such that $1.5 \leq a < b < C$, there is a sum of squares polynomial $p_{a,b,C}(x)$ of degree at most $8\lceil{C^2}\rceil$ such that the following hold:
	    \begin{enumerate}
	        \item For all $x \in [a,b]$, $p_{a,b,C}(x) \geq 1$.
	        \item For all $x \in [0,1]$, $|p_{a,b,C}(x)| \leq \frac{1}{2}$.
	        \item For all $x \in [0,a] \cup [b,C]$, $|p_{a,b,C}(x)| \leq 1$.
	    \end{enumerate}
	\end{lemma}
	\begin{proof}
	    We can take the polynomial 
	    \[
	    p_{a,b,C}(x) = \left(1 - \frac{(x - a)(x-b)}{C^2}\right)^{4\lceil{C^2}\rceil}.
	    \]
	    We now make the following observations:
	    \begin{enumerate}
	        \item For all $x \in [a,b]$, $1 - \frac{(x - a)(x-b)}{C^2} \geq 1$ so  $p_{a,b,C}(x) \geq 1$.
	        \item For all $x \in [0,1]$, $|1 - \frac{(x - a)(x-b)}{C^2}| \leq 1 - \frac{1}{4C^2}$ so $|p_{a,b,C}(x)| \leq \left(1 - \frac{1}{4C^2}\right)^{4\lceil{C^2}\rceil} \leq \frac{1}{2}$.
	        \item For all $x \in [0,a] \cup [b,C]$, $|1 - \frac{(x - a)(x-b)}{C^2}| \leq 1$ so $|p_{a,b,C}(x)| \leq 1$.
	    \end{enumerate}
	\end{proof}
	We construct the polynomial $s_2(x)$. For $k\in \{2,\ldots, \lceil n/2 \rceil\}$, let
	$$
	s_2(x):= p_{a,b,C}\left(H_k(x)\right)^{40\lceil{k \log(n)}\rceil},
	$$
	where $a = H_k(k)$, $b = H_k(k-1)$, and $C = e^{8\sqrt{3}}$ is the constant given by Lemma \ref{lem:H_polynpomial_properties}.
	\begin{lemma}
		\label{lem:h_k_satisties_all_propoerties}
		For any $k\in \{2,\ldots, \lceil n/2 \rceil\}$, $s_2(x)$ satisfies the properties in Lemma~\ref{lem:h_sufficient_cond}. 
	\end{lemma}
	\begin{proof}
	    We make the following observations:
	    \begin{enumerate}
	        \item For all $x \in [0,k-1] \cup [k,2k-1]$, $H_k(x) \in [0,H_k(k)] \cup [H_k(k-1),C]$ so $|p_{a,b,C}(H_k(x))| \leq 1$ and thus
	        $$s_2(x) = p_{a,b,C}\left(H_k(x)\right)^{40\lceil{k \log(n)}\rceil} \leq 1.$$
	        \item For all $x \in [k-1,k]$, $H_k(x) \in [H_k(k),H_k(k-1)]$ so $p_{a,b,C}(H_k(x)) \geq 1$ and thus
	        $$s_2(x) = p_{a,b,C}\left(H_k(x)\right)^{40\lceil{k \log(n)}\rceil} \geq 1.$$
	        \item For all $x \in [2k-1,n]$, $H_k(x) \in [0,1]$ so $|p_{a,b,C}(H_k(x))| \leq 1$ and thus,
	        $$s_2(x) = p_{a,b,C}\left(H_k(x)\right)^{40\lceil{k \log(n)}\rceil} \leq n^{-40k}.$$
	    \end{enumerate}
	\end{proof}

\end{proof}

\section{SoS rank upper bound for the MK problem via SQF certification}\label{section:MK_problem}
In this section, we prove an upper bound of $O(\sqrt{n}\log(P))$ on the SoS rank for the MK problem, which, together with the lower bound presented in~\cite{Kurpisz19}, constitutes proof of Theorem~\ref{thm:MK_SoS_rank}.
We first discuss the necessary properties a candidate SoS certificate for the MK problem has to satisfy.
A degree $d$ SoS certificate for the MK problem is of the form
$\sum_{i\in [n]} x_i-1=s_0(\mathbf{x})+s_1(\mathbf{x}) \left(\sum_{i \in [n]}  x_i - \frac{1}{P} \right), $
where $s_0,s_1$ are SoS polynomials of degree $2d+2$ and $2d$, respectively.
Through permutation of indices, the existence of an SoS certificate for the MK problem implies the existence of an SoS certificate such that $s_1$ is symmetric, that is, there exists $\tilde{s}_1:\R\to\R$ such that ${s_1(\mathbf{x})=\tilde{s}_1(|\mathbf{x}|)}$ for all $\mathbf{x}\in\{0,1\}^n$. Since $s_0$ is globally nonnegative, $\tilde{s}_1$ needs to satisfy 
\begin{equation}\label{eq:necessaryMK}
	|\mathbf{x}|-1  \geq \tilde{s}_1(|\mathbf{x}|)\left(|\mathbf{x}|- \frac{1}{P}\right)  \qquad \text{ for all } \mathbf{x}\in \{0,1\}^n.
\end{equation}
Thus, $\tilde{s}_1(0) \geq P$, $\tilde{s}_1(1)=0$, and $\tilde{s}_1(x)\leq \frac{x-1}{x - \frac{1}{P}}$ for $x\in\{2,\ldots,n\}$.

We will construct a sum of squares polynomial $\tilde{s}_1$ which satisfies the following slightly stronger conditions:
\begin{enumerate}
    \item $\tilde{s}_1(0) > P$.
    \item For all $x \in [1,2]$, $\tilde{s}_1(x) \leq \frac{x-1}{2}$.
    \item For all $x \in [2,n]$, $\tilde{s}_1(x) \leq \frac{1}{2}$.
\end{enumerate}
We will then observe that these conditions imply that 
\[
\tilde{s}_0(|x|) = |x| - 1 - \tilde{s}_1(|x|)\left(|x| - \frac{1}{P}\right)
\]
is positive for all $x \in \{0\} \cup (1,n]$ which is sufficient to show that $\tilde{s}_0(x)$ is a sum of squares modulo the Boolean constraints.

A polynomial $T_{2 \sqrt{n}}(\frac{x-1+r_0}{n}-1)$, where $r_0$ is the smallest root of the polynomial $T_{2 \sqrt{n}}(\frac{x}{n}-1)$, which for $P=2$ satisfies similar 
requirements was constructed in~\cite[Lemma 15]{Kurpisz19} using properties of Chebyshev polynomials.

To obtain our polynomial $\tilde{s}_1(x)$, we generalize this construction using three parameters, the degree $d$ of the Chebyshev polynomial, a scaling factor $\alpha$, and an even power $m$.


\begin{definition}
Given an $\alpha > 0$, a natural number $d$, and an even natural number $m$, define 
$
\tilde{s}_{\alpha,d,m}(x) := {\alpha}T_d\left(\frac{x-1+r_0}{n} - 1\right)^m,
$
where $r_0$ is the smallest root of the polynomial $T_d\left(\frac{x}{n} - 1\right)$.
\end{definition}
\begin{lemma}
    $r_0 \leq \frac{{\pi}^2{n}}{4d^2}$.
\end{lemma}
\begin{proof}
    Observe that $T_d(x) = \cos(d\cos^{-1}(x))$ so the first zero of $T_d(x)$ is $\cos\left(-\pi + \frac{\pi}{2d}\right) \leq -1 + \frac{\pi^2}{4d^2}$. Thus, the first zero of $T_d\left(\frac{x}{n} - 1\right)$ is at most $\frac{{\pi}^2{n}}{4d^2}$.
\end{proof}
\begin{lemma}\label{madpropertieslemma}
For $d > \frac{\pi}{2}\sqrt{n}$ the polynomial $\tilde{s}_{\alpha,d,m}(x)$ satisfies the following properties:
\begin{enumerate}
    \item For all $x \in [1,n]$, $\tilde{s}_{\alpha,d,m}(x) \leq \min{\{\frac{{\alpha}d^2}{n}(x-1),\alpha\}}$.
    \item $\tilde{s}_{\alpha,d,m}(0) \geq \alpha\left(\frac{1}{4}\left(1 + \sqrt{\frac{2(1 - r_0)}{n}}\right)^{d}\right)^{m}$.
\end{enumerate}
\end{lemma}
\begin{proof}
    For the first statement, observe that by the Markov Brothers' Theorem, since $|T_d(x)| \leq 1$ for all $x \in [-1,1]$, $|T'_d(x)| \leq d^2$ for all $x \in [-1,1]$. This implies that $\left|T'_d\left(\frac{x-1+r_0}{n} - 1\right)\right| \leq \frac{d^2}{n}$ for all $x \in [1 - r_0,2n + 1 - r_0]$. Since $T_d\left(\frac{x-1+r_0}{n} - 1\right) = 0$, when $x = 1$, $\left|T_d\left(\frac{x-1+r_0}{n} - 1\right)\right| \leq \min{\{\frac{d^2(x-1)}{n},1\}}$ for all $x \in [1,n]$, which implies the result.
    
    For the second statement, by Lemma \ref{lem:Chebyshev_1-c/n_properties}, if $0 \leq c \leq n$ then $|T_d(-1 - \frac{c}{n})| \geq \frac{1}{4}\left(1 + \sqrt{\frac{2c}{n}}\right)^{d}$. Applying this lemma with $c = 1 - r_0$, the result follows.
\end{proof}
\begin{corollary}\label{madconditionscorollary}
If the conditions
\begin{enumerate}
    \item $d \geq 3\sqrt{n}$,
    \item $\alpha \leq \frac{n}{2d^2} \leq \frac{1}{2}$,
    \item $m > \frac{\ln(P) - ln(\alpha)}{d\ln\left(1 + \sqrt{\frac{2(1 - r_0)}{n}}\right) - \ln(4)}$,
\end{enumerate}
are satisfied,
then the following properties hold:
\begin{enumerate}
    \item $\tilde{s}_{\alpha,d,m}(0) > P$.
    \item For all $x \in [1,2]$, $\tilde{s}_{\alpha,d,m}(x) \leq \frac{x-1}{2}$.
    \item For all $x \in [2,n]$, $\tilde{s}_{\alpha,d,m}(x) \leq \frac{1}{2}$.
\end{enumerate}
Thus,  $(x - 1) - \tilde{s}_{\alpha,d,m}(x)(x - \frac{1}{P}) > 0$
whenever $x \in \{0\} \cup (1,n]$.
\end{corollary}
\begin{proof}
The first statement follows from algebraic manipulations provided that 
$$\frac{1}{4}\left(1 + \sqrt{\frac{2(1 - r_0)}{n}}\right)^{d} \geq 1.$$ To confirm that this holds, observe that $r_0 \leq \frac{{\pi}^2{n}}{4d^2} \leq \frac{1}{2}$. Thus,  
\[
\left(1 + \sqrt{\frac{2(1 - r_0)}{n}}\right)^{d} \geq \left(1 + \frac{1}{\sqrt{n}}\right)^{d} \geq 2^{\frac{d}{\sqrt{n}}} \geq 8.
\]

For the second and third statements, we use the facts that for all $x \in [1,n]$, $\tilde{s}_{\alpha,d,m}(x) \leq \frac{{\alpha}d^2}{n}(x-1)$ and $\tilde{s}_{\alpha,d,m}(x) \leq \alpha$, respectively.

To show that $(x - 1) - \tilde{s}_{\alpha,d,m}(x)(x - \frac{1}{P}) > 0$ whenever $x \in \{0\} \cup (1,n]$, we make the following observations:
\begin{enumerate}
    \item For $x = 0$, $-1 - \tilde{s}_{\alpha,d,m}(0)(-\frac{1}{P}) > -1 - P\left(-\frac{1}{P}\right) = 0$.
    \item For $x \in (1,2]$, $(x - 1) - \tilde{s}_{\alpha,d,m}(x)(x - \frac{1}{P}) \leq (x-1) - \frac{x-1}{2}\left(x - \frac{1}{P}\right) > 0$.
    \item For $x \in [2,n]$, $(x - 1) - \tilde{s}_{\alpha,d,m}(x)(x - \frac{1}{P}) \leq (x-1) - \frac{1}{2}\left(x - \frac{1}{P}\right) > 0$.
\end{enumerate}
\end{proof}
We now confirm that 
\[
\tilde{s}_0 = (x - 1) - \tilde{s}_{\alpha,d,m}(x)\left(x - \frac{1}{P}\right)
\]
is a sum of squares modulo the Boolean axioms. To see this, observe that since $\tilde{s}_0(x) > 0$ for $x \in \{0\} \cup (1,n]$, $\tilde{s}_0(x)$ must have an even number of roots in $(0,1]$ and no other roots in $[0,n]$. Thus, we can write
\[
\tilde{s}_0(x) = p\prod_{i=1}^{l}(x-a_i)(x-b_i)
\]
for some polynomial $p$ which is positive on $[0,n]$ and some real roots $a_1,\ldots,a_l,b_1,\ldots,b_l \in (0,1]$. Since $p$ is positive on $[0,n]$, $p$ is a sum of squares modulo the Boolean axioms. By Corollary  \ref{cor:SQF_SoS_degree_closer_roots_crude_bound}, since $|x|(|x|-1)$ is a sum of squares modulo the Boolean axioms, for each $i \in [l]$, $(x-a_i)(x-b_i)$ is also a sum of squares modulo the Boolean axioms. Thus, $\tilde{s}_0(x)$ is a sum of squares modulo the Boolean axioms.

Finally, we observe that we can satisfy the required conditions on $d$, $\alpha$, and $m$ by taking $d = \lceil{3\sqrt{n}}\rceil$, $\alpha = \frac{1}{2d^2} \approx \frac{1}{18 n}$, and $m = O(\log(P))$, which gives a sum of squares certificate of degree $O(\sqrt{n}\log(P))$.

\section{SoS rank upper bound for the SC problem via SQF certification}\label{sec:SC_Problem}
In this section, we refute the Bienstock---Zuckenberg conjecture for the SC problem. We provide a degree $O(\sqrt{n} \log(n))$ SoS certificate for the SC problem on the Boolean hypercube, thus proving Theorem~\ref{thm:SC_SoS_rank}. For this proof, we use the SoS rank for certifying SQFs for $k=2$ in Theorem~\ref{thm:SQF_SoS_degree}. We present an alternative direct proof in Appendix~\ref{sec:SC_second_proof}.  
We begin this section with a discussion on the properties necessary for an SoS polynomial $s$ to even be considered as a possible candidate for an SoS certificate for the SC problem.
An SoS certificate for the SC problem is of the form
$$\sum_{i\in [n]} x_i-2=s_0(\mathbf{x})+\sum_{i \in [n]}s_i(\mathbf{x})g_i(\mathbf{x}),$$
where $$g_i(\mathbf{x})=\left(\sum_{j \in [n] \setminus\{i\}} x_j - 1 \right).$$
As opposed to the discussion in Section~\ref{section:MK_problem}, an SoS certificate for the SC problem not only has multiple constraints but also displays a certain type of asymmetry, which is present in the formulation of the polynomials $g_i$ for $i\in [n]$. One could hope to abuse this asymmetry by constructing different SoS polynomials $s_i\in \Sigma_{n,d}$ for certain $d\in [n]$, but for this proof, we proceed in a similar fashion as for the MK problem and instead construct only one symmetric SoS polynomial $s:\{0,1\}^n\to\R$ and look for the certificate of the form
$$
\sum_{i\in [n]} x_i-2=s_0(\mathbf{x})+\sum_{i \in [n]}s(\mathbf{x})g_i(\mathbf{x}).
$$
Through permutation of indices, the existence of an SoS certificate for the SC problem implies the existence of an SoS certificate such that $s$ is symmetric, that is, there exists an $\tilde{s}:\R\to\R$ such that ${s(\mathbf{x})=\tilde{s}(|\mathbf{x}|)}$ for all $\mathbf{x}\in\{0,1\}^n$. As for the MK problem, we are interested in the requirements that polynomial $\tilde{s}$ needs to satisfy such that $s$ constitutes part of an SoS certificate for the SC problem. Let $g(\mathbf{x}):= \sum_{i \in [n]}g_i(\mathbf{x})=(n-1)(\sum_{i=1}^n x_i)-n$ and note that $g$ is a symmetric polynomial; there exists a univariate polynomial $\tilde{g}$ such that $\tilde{g}(|\textbf{x}|)= g(\textbf{x})$ for all $\mathbf{x}\in\{0,1\}^n$. Since $s_0$ is globally nonnegative, this implies that $s$ needs to satisfy
\begin{align}\label{ineq:SCCertificate}
	|\mathbf{x}|-2&\geq \tilde{s}(|\mathbf{x}|)\left(|\mathbf{x}|(|\mathbf{x}|-2)+(n-|\mathbf{x}|)(|\mathbf{x}|-1)\right)\nonumber \\
	&=\tilde{s}(|\mathbf{x}|)((n-1)|\mathbf{x}|-n)= \tilde{s}(|\mathbf{x}|)\tilde{g}(|\mathbf{x}|) \qquad \text{for all } \textbf{x}\in \{0, 1\}^n.
\end{align}
This implies that $\tilde{s}(0) \geq \frac{2}{n}$, $\tilde{s}(1) \geq 1$, $\tilde{s}(2) = 0$ and $\tilde{s}(x) \leq \frac{x-2}{3(n-1)x - n}$ for all $x \in \{3,4,\ldots,n\}$. We will construct a sum of squares polynomial $\tilde{s}(x)$ which satisfies the following slightly stronger conditions:
\begin{enumerate}
    \item $\tilde{s}(x) \geq 1$ for all $x \in [0,1]$.
    \item For all $x \in [1,2)$, $\frac{\tilde{s}(x)}{x-2} < 0$ and $\frac{\tilde{s}(x)}{x-2}$ is increasing.
    \item $\tilde{s}(x) \leq \frac{(x-2)}{2n}$ for all $x \in [2,3]$.
    \item $\tilde{s}(x) \leq \frac{1}{2n}$ for all $x \in [3,n]$.
\end{enumerate}
We will then observe that these conditions imply that $\tilde{s}_0(x) = x - 2 - \tilde{s}(x)((n-1)x-n)$ is positive for $x \in [0,1) \cup (2,n]$ and has exactly two zeros in the interval $[1,2]$, one of which is $x = 2$. We can then use Theorem \ref{thm:SQF_SoS_degree} and Corollary \ref{cor:SQF_SoS_degree_closer_roots_crude_bound} to show that $\tilde{s}_0$ is a sum of squares of degree $\deg(\tilde{s}) + O(\sqrt{n}\log(n))$ modulo the Boolean axioms.
\begin{lemma}
    For $d = 3\sqrt{n}$, $\alpha = \frac{1}{18n}$, and $m = 2\lceil{\log_2(\sqrt{18n})}\rceil$ the polynomial $\tilde{s}(x) = \tilde{s}_{\alpha,d,m}(x-1)$ satisfies the following properties: 
    \begin{enumerate}
        \item $\tilde{s}(x) \geq 1$ for all $x \in [0,1]$.
        \item For all $x \in [1,2)$, $\frac{\tilde{s}(x)}{x-2} < 0$ and $\frac{\tilde{s}(x)}{x-2}$ is increasing.
        \item $\tilde{s}(x) \leq \frac{(x-2)}{2n}$ for all $x \in [2,3]$.
        \item $\tilde{s}(x) \leq \frac{1}{2n}$ for all $x \in [3,n]$.
    \end{enumerate}
\end{lemma}
\begin{proof}
    For the first statement, just as in the proof of Corollary~\ref{madconditionscorollary}, $r_0 \leq \frac{{\pi}^2{n}}{4d^2} \leq \frac{1}{2}$. Thus,
    \[
    \left(1 + \sqrt{\frac{2(1 - r_0)}{n}}\right)^{d} \geq \left(1 + \frac{1}{\sqrt{n}}\right)^{d} \geq 2^{\frac{d}{\sqrt{n}}} \geq 8.
    \]
    Hence, by Lemma~\ref{madpropertieslemma}, $\tilde{s}(1) = \tilde{s}_{\alpha,d,m}(0) \geq {\alpha}2^{m} \geq 1$. Since $deg(\tilde{s})$ is even, all roots of $\tilde{s}$ are real and the smallest root of $\tilde{s}$ is $2$, $\tilde{s}$ is positive and decreasing when $x < 2$ so $\tilde{s}(x) \geq 1$ whenever $x \in [0,1]$, as needed.
    
    For the second statement, observe that since $\deg(\tilde{s})$ is even, all roots of $\tilde{s}$ are real and the smallest root of $\tilde{s}$ is $2$, $\frac{\tilde{s}(x)}{x-2}$ is negative and increasing whenever $x < 2$.
    
    For the third statement, observe that by Lemma \ref{madpropertieslemma}, for all $x \in [2,3]$, $\tilde{s}(x) = \tilde{s}_{\alpha,d,m}(x-1) \leq {\alpha}\frac{d^2}{n}(x-2) \leq \frac{x-2}{2n}$.
    
    For the fourth statement, observe that by Lemma \ref{madpropertieslemma}, for all $x \in [3,n]$, $\tilde{s}(x) = \tilde{s}_{\alpha,d,m}(x-1) \leq {\alpha} < \frac{1}{2n}$.
\end{proof}
\begin{corollary}
    For $d = 3\sqrt{n}$, $\alpha = \frac{1}{n}$, $m = 2\lceil{\log_2(n)}\rceil$, and $\tilde{s}(x) = \tilde{s}_{\alpha,d,m}(x-1)$ the polynomial 
    $\tilde{s}_0(x) = x - 2 - \tilde{s}(x)((n-1)x-n)$ is positive for $x \in [0,1) \cup (2,n]$ and has exactly two zeros in the interval $[1,2]$, one of which is $x = 2$.
\end{corollary}
\begin{proof}
    We make the following observations:
    \begin{enumerate}
        \item For all $x \in [0,1)$, 
        \[
        \tilde{s}_0(x) = x - 2 - \tilde{s}(x)((n-1)x-n) \geq x - 2 - ((n-1)x-n) = (n-2)(1-x) > 0.
        \]
        \item For all $x \in [1,2]$, $\frac{\tilde{s}_0}{x-2} = 1 - ((n-1)x-n)\frac{\tilde{s}}{x-2}$. When $x \in [\frac{n}{n-1},2]$, $((n-1)x-n)\frac{\tilde{s}}{x-2} \leq 0$ so $\frac{\tilde{s}}{x-2} > 0$. When $x \in [1,\frac{n}{n-1})$, both $((n-1)x-n)$ and $\frac{\tilde{s}}{x-2}$ are negative and increasing so $((n-1)x-n)\frac{\tilde{s}}{x-2}$ is positive and decreasing and thus $\frac{\tilde{s}_0}{x-2}$ is increasing. Since $\frac{\tilde{s}_0(1)}{1-2} \leq 0$ and $\frac{\tilde{s}_0(\frac{n}{n-1})}{\frac{n}{n-1}-2} > 0$, $\frac{\tilde{s}_0(x)}{x-2}$ must have exactly one zero in the interval $[1,\frac{n}{n-1}]$.
        \item For all $x \in (2,3]$, $\tilde{s}_0(x) = x - 2 - \tilde{s}(x)((n-1)x-n) \geq x-2 - \frac{(n-1)x-n}{2n}(x-2) > 0$.
        \item For all $x \in [3,n]$, $\tilde{s}_0(x) = x - 2 - \tilde{s}(x)((n-1)x-n) \geq x-2 - \frac{(n-1)x-n}{2n} > \frac{x}{2} - \frac{3}{2} \geq 0$.
    \end{enumerate}
\end{proof}
\begin{corollary}
    $\tilde{s}_0(|x|)$ is a sum of squares of degree $O(\sqrt{n}\log(n))$ modulo the Boolean axioms.
\end{corollary}
\begin{proof}
    Since $\tilde{s}_0(x) = x - 2 - \tilde{s}(x)((n-1)x-n)$ is positive for $x \in [0,1) \cup (2,n]$ and has exactly two zeros in the interval $[1,2]$, one of which is $x = 2$, we can write 
    \[
	\tilde{s}_{0}(x)= \tilde{p}(x-a)(x-2),
	\]
	for some $a \in [1,2)$ where $\tilde{p}(x)$ is positive for has no real roots in the interval $[0,n]$. Since $\tilde{p}(x)$ is positive and has no real roots in the interval $[0,n]$, $\tilde{p}(|x|)$ is a sum of squares modulo the Boolean axioms. By Theorem \ref{thm:SQF_SoS_degree} and Corollary \ref{cor:SQF_SoS_degree_closer_roots_crude_bound}, $(x-a)(x-2)$ is a sum of squares of degree $O(\sqrt{n}\log(n))$ modulo the Boolean axioms. 
\end{proof}
Thus, there exists a degree $O(\sqrt{n}\log (n))$ SoS certificate of nonnegativity for the SC problem.

{\small
\bibliographystyle{abbrv}
\bibliography{bibliography}

\begin{thebibliography}{10}

\bibitem{AroraRV09}
S.~Arora, S.~Rao, and U.~V. Vazirani.
\newblock Expander flows, geometric embeddings and graph partitioning.
\newblock {\em J. {ACM}}, 56(2):5:1--5:37, 2009.

\bibitem{BarakHKKMP16}
B.~Barak, S.~B. Hopkins, J.~A. Kelner, P.~Kothari, A.~Moitra, and A.~Potechin.
\newblock A nearly tight sum-of-squares lower bound for the planted clique
  problem.
\newblock In {\em {IEEE} 57th Annual Symposium on Foundations of Computer
  Science, {FOCS} 2016, 9-11 October 2016, Hyatt Regency, New Brunswick, New
  Jersey, {USA}}, pages 428--437, 2016.

\bibitem{BarakKS15}
B.~Barak, J.~A. Kelner, and D.~Steurer.
\newblock Dictionary learning and tensor decomposition via the sum-of-squares
  method.
\newblock In {\em {STOC} 2015, Portland, OR, USA, June 14-17, 2015}, pages
  143--151, 2015.

\bibitem{BarakM16}
B.~Barak and A.~Moitra.
\newblock Noisy tensor completion via the sum-of-squares hierarchy.
\newblock In {\em {COLT} 2016, New York, USA, June 23-26, 2016}, pages
  417--445, 2016.

\bibitem{BarakRS11}
B.~Barak, P.~Raghavendra, and D.~Steurer.
\newblock Rounding semidefinite programming hierarchies via global correlation.
\newblock In {\em FOCS}, pages 472--481, 2011.

\bibitem{BarakS16}
B.~Barak and D.~Steurer.
\newblock Proofs, beliefs, and algorithms through the lens of sum-of-squares,
  2016.
\newblock https://www.sumofsquares.org.

\bibitem{BhaskaraCVGZ12}
A.~Bhaskara, M.~Charikar, A.~Vijayaraghavan, V.~Guruswami, and Y.~Zhou.
\newblock Polynomial integrality gaps for strong sdp relaxations of densest
  {\it k}-subgraph.
\newblock In {\em SODA}, 2012.

\bibitem{BienstockZ04}
D.~Bienstock and M.~Zuckerberg.
\newblock Subset algebra lift operators for 0-1 integer programming (extended
  version).
\newblock 2002.
\newblock Extended version of: Subset algebra lift operators for 0-1 integer
  programming - SIAM Journal on Optimization, 1(15): 63-95, 2004.

\bibitem{blekherman2012semidefinite}
G.~Blekherman, P.~A. Parrilo, and R.~R. Thomas.
\newblock {\em Semidefinite optimization and convex algebraic geometry}.
\newblock SIAM, 2012.

\bibitem{Cheung07}
K.~K.~H. Cheung.
\newblock Computation of the {L}asserre ranks of some polytopes.
\newblock {\em Math. Oper. Res.}, 32(1):88--94, 2007.

\bibitem{cook2001matrix}
W.~Cook and S.~Dash.
\newblock On the matrix-cut rank of polyhedra.
\newblock {\em Math. Oper. Res.}, 26(1):19--30, 2001.

\bibitem{FawziSP15}
H.~Fawzi, J.~Saunderson, and P.~A. Parrilo.
\newblock Sparse sum-of-squares certificates on finite abelian groups.
\newblock In {\em 54th {IEEE} Conference on Decision and Control, {CDC} 2015,
  Osaka, Japan, December 15-18, 2015}, pages 5909--5914, 2015.

\bibitem{GhoshJJPR20}
M.~Ghosh, F.~G. Jeronimo, C.~Jones, A.~Potechin, and G.~Rajendran.
\newblock Sum-of-squares lower bounds for sherrington-kirkpatrick via planted
  affine planes.
\newblock In {\em 61st {IEEE} Annual Symposium on Foundations of Computer
  Science, {FOCS}}, pages 954--965, 2020.

\bibitem{GoemansW95}
M.~X. Goemans and D.~P. Williamson.
\newblock Improved approximation algorithms for maximum cut and satisfiability
  problems using semidefinite programming.
\newblock {\em J. Assoc. Comput. Mach.}, 42(6):1115--1145, 1995.

\bibitem{Grigoriev01}
D.~Grigoriev.
\newblock Complexity of positivstellensatz proofs for the knapsack.
\newblock {\em Comput. Complexity}, 10(2):139--154, 2001.

\bibitem{Grigoriev01b}
D.~Grigoriev.
\newblock Linear lower bound on degrees of positivstellensatz calculus proofs
  for the parity.
\newblock {\em Theoretical Computer Science}, 259(1-2):613--622, 2001.

\bibitem{GrigorievHP02}
D.~Grigoriev, E.~A. Hirsch, and D.~V. Pasechnik.
\newblock Complexity of semi-algebraic proofs.
\newblock In {\em STACS}, pages 419--430, 2002.

\bibitem{GrigorievV01}
D.~Grigoriev and N.~Vorobjov.
\newblock Complexity of null-and positivstellensatz proofs.
\newblock {\em Ann. Pure App. Logic}, 113(1-3):153--160, 2001.

\bibitem{GuruswamiS11}
V.~Guruswami and A.~K. Sinop.
\newblock {L}asserre hierarchy, higher eigenvalues, and approximation schemes
  for graph partitioning and quadratic integer programming with psd objectives.
\newblock In {\em FOCS}, pages 482--491, 2011.

\bibitem{HopkinsSSS16}
S.~B. Hopkins, T.~Schramm, J.~Shi, and D.~Steurer.
\newblock Fast spectral algorithms from sum-of-squares proofs: tensor
  decomposition and planted sparse vectors.
\newblock In {\em {STOC} 2016, Cambridge, MA, USA, June 18-21, 2016}, pages
  178--191, 2016.

\bibitem{KothariSS18}
P.~Kothari, J.~Steinhardt, and D.~Steurer.
\newblock Robust moment estimation and improved clustering via sum of squares.
\newblock In {\em {STOC} 2018}, 2018.

\bibitem{KothariMOW17}
P.~K. Kothari, R.~Mori, R.~O'Donnell, and D.~Witmer.
\newblock Sum of squares lower bounds for refuting any {CSP}.
\newblock In {\em Proceedings of the 49th Annual {ACM} {SIGACT} Symposium on
  Theory of Computing, {STOC} 2017, Montreal, QC, Canada, June 19-23, 2017},
  pages 132--145, 2017.

\bibitem{BandeiraK19}
D.~Kunisky and A.~S. Bandeira.
\newblock A tight degree 4 sum-of-squares lower bound for the
  sherrington-kirkpatrick hamiltonian.
\newblock {\em CoRR}, abs/1907.11686, 2019.

\bibitem{Kurpisz19}
A.~Kurpisz.
\newblock Sum-of-squares bounds via boolean function analysis.
\newblock In {\em {ICALP} July 9-12, 2019, Patras, Greece}, 2019.

\bibitem{KurpiszLM16}
A.~Kurpisz, S.~Lepp{\"{a}}nen, and M.~Mastrolilli.
\newblock Sum-of-squares hierarchy lower bounds for symmetric formulations.
\newblock In {\em Integer Programming and Combinatorial Optimization - 18th
  International Conference, {IPCO} 2016, Li{\`{e}}ge, Belgium, June 1-3, 2016,
  Proceedings}, pages 362--374, 2016.

\bibitem{KurpiszLM16c}
A.~Kurpisz, S.~Lepp{\"{a}}nen, and M.~Mastrolilli.
\newblock Tight sum-of-squares lower bounds for binary polynomial optimization
  problems.
\newblock In {\em 43rd International Colloquium on Automata, Languages, and
  Programming, {ICALP} 2016, July 11-15, 2016, Rome, Italy}, pages 78:1--78:14,
  2016.

\bibitem{KurpiszLM17}
A.~Kurpisz, S.~Lepp{\"{a}}nen, and M.~Mastrolilli.
\newblock On the hardest problem formulations for the 0/1 lasserre hierarchy.
\newblock {\em Math. Oper. Res.}, 42(1):135--143, 2017.

\bibitem{KurpiszLM17b}
A.~Kurpisz, S.~Lepp{\"{a}}nen, and M.~Mastrolilli.
\newblock An unbounded sum-of-squares hierarchy integrality gap for a
  polynomially solvable problem.
\newblock {\em Math. Program.}, 166(1-2):1--17, 2017.

\bibitem{Lasserre01z}
J.~B. Lasserre.
\newblock An explicit exact {SDP} relaxation for nonlinear 0-1 programs.
\newblock In {\em Integer Programming and Combinatorial Optimization, 8th
  International {IPCO} Conference, Utrecht, The Netherlands, June 13-15, 2001,
  Proceedings}, pages 293--303, 2001.

\bibitem{Laurent03}
M.~Laurent.
\newblock A comparison of the {S}herali-{A}dams, {L}ov{\'a}sz-{S}chrijver, and
  {L}asserre relaxations for 0-1 programming.
\newblock {\em Math. Oper. Res.}, 28(3):470--496, 2003.

\bibitem{Laurent03a}
M.~Laurent.
\newblock Lower bound for the number of iterations in semidefinite hierarchies
  for the cut polytope.
\newblock {\em Math. Oper. Res.}, 28(4):871--883, 2003.

\bibitem{LeeRagSteu15}
J.~R. Lee, P.~Raghavendra, and D.~Steurer.
\newblock Lower bounds on the size of semidefinite programming relaxations.
\newblock In {\em STOC}, pages 567--576, 2015.

\bibitem{LeePWY16}
T.~Lee, A.~Prakash, R.~Wolf, and H.~Yuen.
\newblock On the sum-of-squares degree of symmetric quadratic functions.
\newblock In {\em 31st Conference on Computational Complexity, {CCC} 2016, May
  29 to June 1, 2016, Tokyo, Japan}, pages 17:1--17:31, 2016.

\bibitem{Lovasz79}
L.~Lov\'asz.
\newblock On the shannon capacity of a graph.
\newblock {\em IEEE Transactions on Information Theory}, 25:1--7, 1979.

\bibitem{MekaPW15}
R.~Meka, A.~Potechin, and A.~Wigderson.
\newblock Sum-of-squares lower bounds for planted clique.
\newblock In {\em Proceedings of the 47th Annual {ACM} on Symposium on Theory
  of Computing, {STOC} 2015, Portland, OR, USA}, pages 87--96, 2015.

\bibitem{Moore68}
M.~J. Moore.
\newblock An $n$ job, one machine sequencing algorithm for minimizing the
  number of late jobs.
\newblock {\em Management Science}, 15:102--109, 1968.

\bibitem{MurtyK87}
K.~G. Murty and S.~N. Kabadi.
\newblock Some np-complete problems in quadratic and nonlinear programming.
\newblock {\em Mathematical Programming}, 39(2):117--129, 1987.

\bibitem{Nesterov00}
Y.~Nesterov.
\newblock {\em Global quadratic optimization via conic relaxation}, pages
  363--384.
\newblock Kluwer Academic Publishers, 2000.

\bibitem{parrilo00}
P.~Parrilo.
\newblock {\em Structured Semidefinite Programs and Semialgebraic Geometry
  Methods in Robustness and Optimization}.
\newblock {PhD} thesis, California Institute of Technology, 2000.

\bibitem{Paturi92}
R.~Paturi.
\newblock On the degree of polynomials that approximate symmetric boolean
  functions (preliminary version).
\newblock In {\em Proceedings of the 24th Annual {ACM} Symposium on Theory of
  Computing, May 4-6, 1992, Victoria, British Columbia, Canada}, pages
  468--474, 1992.

\bibitem{potechin2020sum}
A.~Potechin.
\newblock Sum of squares bounds for the ordering principle.
\newblock In {\em Proceedings of the 35th Computational Complexity Conference},
  pages 1--37, 2020.

\bibitem{PotechinS17}
A.~Potechin and D.~Steurer.
\newblock Exact tensor completion with sum-of-squares.
\newblock In {\em {COLT} 2017, Amsterdam, The Netherlands, 7-10 July 2017},
  pages 1619--1673, 2017.

\bibitem{Rivlin74}
T.~Rivlin.
\newblock The chebyshev polynomials.
\newblock {\em SERBIULA (sistema Librum 2.0)}, 02 1974.

\bibitem{SakueTKI17}
S.~Sakaue, A.~Takeda, S.~Kim, and N.~Ito.
\newblock Exact semidefinite programming relaxations with truncated moment
  matrix for binary polynomial optimization problems.
\newblock {\em SIAM Journal on Optimization}, 27(1):565--582, 2017.

\bibitem{SchrammS17}
T.~Schramm and D.~Steurer.
\newblock Fast and robust tensor decomposition with applications to dictionary
  learning.
\newblock In {\em {COLT} 2017, Amsterdam, The Netherlands, 7-10 July 2017},
  pages 1760--1793, 2017.

\bibitem{schor87}
N.~Shor.
\newblock Class of global minimum bounds of polynomial functions.
\newblock {\em Cybernetics}, 23(6):731--734, 1987.

\bibitem{SlotL19}
L.~Slot and M.~Laurent.
\newblock Improved convergence analysis of lasserre's measure-based upper
  bounds for polynomial minimization on compact sets.
\newblock 2019.

\bibitem{ThapperZ17}
J.~Thapper and S.~Zivny.
\newblock The power of sherali-adams relaxations for general-valued csps.
\newblock {\em {SIAM} J. Comput.}, 46(4):1241--1279, 2017.

\bibitem{Tulsiani09}
M.~Tulsiani.
\newblock Csp gaps and reductions in the lasserre hierarchy.
\newblock In {\em STOC}, pages 303--312, 2009.

\bibitem{Wolf10}
R.~Wolf.
\newblock A note on quantum algorithms and the minimal degree of
  {\(\epsilon\)}-error polynomials for symmetric functions.
\newblock {\em Quantum Information {\&} Computation}, 8(10):943--950, 2010.

\end{thebibliography}
}

\appendix

\section{Alternative Proof for the SoS rank upper bound for the SC problem}
\label{sec:SC_second_proof}
In this section, we 
provide an alternative proof of Theorem~\ref{thm:SC_SoS_rank}. More precisely, we prove an $O(\sqrt{n}\log(n) )$ upper bound on the SoS rank for the SC problem without using Theorem \ref{thm:SQF_SoS_degree}.

By the problem formulation, Definition~\eqref{eq:MK_def}, and  Equation~\eqref{eq:intro_SoS_d_CPOP}, the SoS rank for the SC Problem is the smallest $d$ for which there exist SoS polynomials $s_0 \in \Sigma_{n,2d+2}$ and $s_i \in \Sigma_{n,2d}$ for $i \in[n]$ such that 
$$
\sum_{i=1}^n x_i -2 = s_0(\mathbf{x}) + \sum_{i=1}^n s_i \left(\sum_{\substack{j=1 \\ j \neq i}}^n x_j -1\right).
$$
Equivalently, it is the smallest positive integer $d$ such that $\sum_{i=1}^n x_i -2 \in \Sigma_{n,d}^\mathcal{G}$.

To prove the SoS rank upper bound for the SC problem, we define the polynomials
$$
	h_1(\mathbf{x}):=|\mathbf{x}| -1
$$
	and
$$	
    h_2(\mathbf{x}):=|\mathbf{x}|\left(|\mathbf{x}| -2\right)
$$
and require the following lemma, in which we use the asymmetry inherent to the constraints of the SC problem.
\begin{lemma}
	\label{lem:two_polnomials_in_sigma}
	For polynomials $h_1,h_2$ it holds that $h_1(\mathbf{x})\in \Sigma_{n,0}^\mathcal{G}$ and $ h_2(\mathbf{x})\in \Sigma_{n, 1}^\mathcal{G}$.
\end{lemma}
\begin{proof}
	Consider the first polynomial, $h_1$, and note that
	$$\sum_{i=1}^n x_i -1 =\frac{1}{n-1} \sum_{i=1}^n \left(\sum_{\substack{j=1 \\ j \neq i}}^n x_j -1\right) + \frac{1}{n-1}\in \Sigma_{n,0}^\mathcal{G}.$$
	Polynomial $h_2$ can be written as
	{\small
	\begin{align*}
		\sum_{j=1}^n x_j\left(\sum_{i=1}^n x_i -2\right) &= 
		\sum_{j=1}^n \left( x_j \left( \sum_{i=1}^n x_i -x_j -1 \right) +(x_j^2-x_j) \right)\\
		&=\sum_{j=1}^n \left( x_j^2 \left( \sum_{\substack{i=1 \\ i \neq j}}^n x_i -1 \right) - (x_j^2 -x_j) \left( \sum_{\substack{i=1 \\ i \neq j}}^n x_i -1 \right)+(x_j^2-x_j) \right) \in \Sigma_{n, 1}^\cG.
	\end{align*}}
\end{proof}
Although Lemma~\ref{lem:two_polnomials_in_sigma} uses asymmetry in the constraints of the SC problem, both $h_1$ and $h_2$ are symmetric polynomials. We can thus define polynomials $\tilde{h}_1,~\tilde{h}_2:~\mathbb{R} \rightarrow \mathbb{R}$ such that $\tilde{h}_1(|\mathbf{x}|)= h_1(\mathbf{x})$, and $\tilde{h}_2(|\mathbf{x}|)= h_2(\mathbf{x})$, respectively. We are working towards a proof of the existence of polynomials $p_1, p_2:~\mathbb{R} \rightarrow \mathbb{R}$ such that
\begin{equation}
	\label{eq:SC_certificate_univariate}
	(x-2)- p_1(x) \tilde{h}_1(x)-p_2(x)\tilde{h}_2(x) \geq 0 \qquad \qquad \text{for all } x \in [0,n].
\end{equation}

\subsection{Construction of polynomials $\mathbf{p_1,p_2}$}

We consider necessary, but not sufficient requirements that the polynomials $p_1$ and $p_2$ have to satisfy, that is, $p_1(2)\tilde{h}_1(2) + p_2(2)\tilde{h}_2(2)=0$, $\left[p_1\tilde{h}_1+p_2\tilde{h}_2\right]^{'}(2)=1$, and $\left[p_1\tilde{h}_1+p_2\tilde{h}_2\right]^{''}(2)<0$. It is easy to check that these requirements are satisfied if $p_1$ has a double root at $x=2$, $p_2(2)=1/2$, and $1 + 4 p^{'}_2(2) +2 \frac{p_1(2)}{(x-2)^2} <0$. We use these guidelines to construct polynomials 
\begin{align}
	\label{eq:Cheby_1}
	p_1(x)&:=\frac{1}{2n^2c_1}~(x-2)^2~ T^2_{2\sqrt{n}\log(n)}\left(2 \frac{x-2}{n}-1 \right),\\
	p_2(x)&:=\frac{1}{2nc_2}~T^2_{2\sqrt{n}\log(n)}\left(2\frac{x-3}{n}-1 \right),\nonumber
\end{align}
where $c_1$ and $c_2$ are constants equal to $\frac{1}{2n^2}T^2_{2\sqrt{n}\log(n)}\left(-\frac{2}{n}-1 \right)$ and $\frac{1}{n}~T^2_{2\sqrt{n}\log(n)}\left(-\frac{2}{n}-1 \right)$, respectively, such that $p_1(1)=1$ and $p_2(2)=1/2$.
\begin{lemma}
	\label{lem:p1_properties}
	There exists $C \in \mathbb{N}$ such that for $n \geq C$, the polynomial $p_1$ satisfies the following properties:
	\begin{enumerate}
		\item $p_1(x) \geq 4$ \qquad for $x\in [0,\frac{1}{2}]$. \label{prop:1_p1}
		\item $p_1(x) \leq \left(-0.9(x-1)+1 \right) (x-2)^2$ \qquad for $x\in [1,2]$.\label{prop:2_p1}
		\item $p_1(x) \leq \frac{1}{2n^2}(x-2)^2$ \qquad for $x \in [2,n]$. \label{prop:3_p1}
	\end{enumerate}
\end{lemma}
\begin{proof}
	Since $p_1(x)$ is decreasing for $x \leq 1$, to prove Property~\eqref{prop:1_p1} it is enough to show that $p_1(\frac{1}{2})~\geq~4$. By Lemma~\ref{lem:Chebyshev_1-c/n_properties} and for sufficiently big $n$, it holds
	$$
	p_1(1/2)=\frac{~\frac{9}{4}~ T^2_{2\sqrt{n}\log(n)}\left(\frac{-3}{n}-1\right)}{~ T^2_{2\sqrt{n}\log(n)}\left(-\frac{2}{n}-1\right)}  
	\geq
	\frac{1}{2}\left(\frac{1+\sqrt{\frac{6}{n}}}{1+\sqrt{\frac{5}{n}}}\right)^{4\sqrt{n} \log(n)}.
	$$
	Since $\frac{1}{2}\left(\frac{1+\sqrt{\frac{6}{n}}}{1+\sqrt{\frac{5}{n}}}\right)^{4\sqrt{n} \log(n)} \geq 4$ for $n\geq 32$ and by monotonicity, Property~\eqref{prop:1_p1} is satisfied.
	
	To prove Property~\eqref{prop:3_p1}, note that for every $x \in [2,n]$ and $d \in \mathbb{N}$ we have $T^2_d(2\frac{x-2}{n}-1) \leq 1$ and for every $n \geq 2$, by Lemma~\ref{lem:Chebyshev_1-c/n_properties}, we have
	$$c_1=\frac{1}{2n^2}~ T^2_{2\sqrt{n}\log(n)}\left(-\frac{2}{n}-1\right) \geq 
	\frac{1}{2n^2}\frac{1}{4}\left(1+\sqrt{\frac{4}{n}}\right)^{4\sqrt{n}\log(n)}  \geq 1.$$
	
	To prove Property~\eqref{prop:2_p1}, we show that $\frac{1}{2n^2}T^2_{2\sqrt{n}\log(n)}\left(2 \frac{x-2}{n}-1 \right) \leq \left(-0.9(x-1)+1 \right)$ for every $x\in [1,2]$. By construction, it is satisfied for $x=1$ and by Property~\eqref{prop:3_p1}, it is satisfied for $x=2$. Since the function $T^2_{2\sqrt{n}\log(n)}\left(2 \frac{x-2}{n}-1 \right)$ is convex in the interval $[1,2]$, the property is satisfied for $x\in [1,2]$.
\end{proof}

\begin{lemma}
	\label{lem:p2_properties}
	There exists a constant $C \in \mathbb{N}$ such that for $n \geq C$, the polynomial $p_2$ satisfies the following properties:
	\begin{enumerate}
		\item $p_2(x) \geq 4$ \qquad for $x\in [0,1]$. \label{prop:1_p2}
		\item $p_2(2) =\frac{1}{2}$.\label{prop:2_p2}
		\item $p^{'}_2(x)\leq -1$ \qquad for $x\in [1,2]$. \label{prop:3_p2}
		\item $p_2 (x) \leq -0.45 (x-2) +\frac{1}{2}$ \qquad for $x\in [2,3]$. \label{prop:4_p2}
		\item $p_2(x) \leq \frac{1}{2n}$ \qquad for $x \in [3,n]$. \label{prop:5_p2}
	\end{enumerate}
\end{lemma}
\begin{proof}
	Since $p_2(x)$ is decreasing for $x \leq 1$, to prove Property~\eqref{prop:1_p1}, it is enough to show that $p_2(1)~\geq~4$. For sufficiently big $n$ we get:
	$$
	p_2(1):=\frac{\frac{1}{2}~T^2_{2\sqrt{n}\log(n)}\left(-\frac{4}{n}-1\right)}{~T^2_{2\sqrt{n}\log(n)}\left(-\frac{2}{n}-1\right)} \geq
	 \frac{1}{8}  \left(\frac{ -1-\sqrt{\frac{8}{n}}   }{ -1-\sqrt{\frac{5}{n}}  }\right)^{4 \sqrt{n}\log(n)}.
	$$
	Since
	$$\frac{1}{8}\left(\frac{ -1-\sqrt{\frac{8}{n}}   }{ -1-\sqrt{\frac{5}{n}}  }\right)^{4 \sqrt{n}\log(n)}\geq 4$$
	for $n\geq 13$ and by monotonicity, Property~\eqref{prop:1_p1} is satisfied.
	
	Property~\eqref{prop:2_p2} is satisfied by construction.
	
	Since for every $x \in [3,n]$ and $d \in \N$ we have $\left|{T_d(2\frac{x-3}{n}-1)}\right| \leq 1$ and for every $n \geq 2$, by Lemma~\ref{lem:Chebyshev_1-c/n_properties},
	we have
	$$c_2=\frac{1}{n}~ T^2_{2\sqrt{n}\log(n)}\left(-\frac{2}{n}-1\right) \geq \frac{1}{n}\frac{1}{4} \left(-1-\sqrt{\frac{4}{n}} \right)^{4\sqrt{n}\log(n)}  \geq  1,$$
	Property~\eqref{prop:5_p2} is satisfied.
	
	Since $p_2(x)$ is convex for $x \in [1,2]$, to prove Property~\eqref{prop:3_p2}, it is enough to show $p_2^{'}(2)\leq -1$. Note that $\frac{\partial T_d(x)}{\partial x}= dU_{d-1}(x)$ and $\frac{\partial T_d^2(x)}{\partial x}= 2dT_d(x)U_{d-1}(x)$, where $U_d(x)$ is a Chebyshev polynomial of the second type. Thus,
	$$
	p_2^{'}(x)=\frac{4 \log (n) T_{2 \sqrt{n} \log (n)}\left(\frac{2 (x-3)}{n}-1\right) U_{2 \sqrt{n} \log (n)-1}\left(\frac{2 (x-3)}{n}-1\right)}{\sqrt{n} T_{2 \sqrt{n} \log (n)}\left(-1-\frac{2}{n}\right)},
	$$
	which implies that
	$$
	p^{'}_2(2)=\frac{4 \log (n) U_{2 \sqrt{n} \log (n)-1}\left(-1-\frac{2}{n}\right)}{\sqrt{n} T_{2 \sqrt{n} \log (n)}\left(-1-\frac{2}{n}\right){}} =
	\left[  \frac{\partial }{\partial x} \frac{T_{2 \sqrt{n} \log (n)}\left(2\frac{x-3}{n}-1\right)}{T_{2 \sqrt{n} \log (n)}\left(-1-\frac{2}{n}\right)}  \right]\left(2\right).
	$$
	Since $T_{2 \sqrt{n} \log (n)}\left(2\frac{x-3}{n}-1\right)$ for $x=2.5$ takes at most half of the value for $x=2$ and since $T_{2 \sqrt{n} \log (n)}\left(2\frac{x-3}{n}-1\right)$ is convex in the interval $[2,3]$, $p_2^{'}(x) \leq -1$.  By Lemma~\ref{lem:Chebyshev_1-c/n_properties},
	$$
	\frac{T^2_{2\sqrt{n}\log(n)}\left(-\frac{2}{n}-1\right)}{~T^2_{2\sqrt{n}\log(n)}\left(-\frac{1}{n}-1\right)} \geq
	\frac{1}{4}  \left(\frac{ -1-\sqrt{\frac{4}{n}}   }{ -1-\sqrt{\frac{3}{n}}  }\right)^{2 \sqrt{n}\log(n)}.
	$$
	Since
	$$\frac{1}{4}\left(\frac{ -1-\sqrt{\frac{4}{n}}   }{ -1-\sqrt{\frac{3}{n}}  }\right)^{2 \sqrt{n}\log(n)} \geq 2$$
	for $n\geq 100$ and by monotonicity, Property~\eqref{prop:3_p2} is satisfied.

	By Property~\eqref{prop:2_p2}, Property~\eqref{prop:4_p2} holds for $x=2$. By Property~\eqref{prop:5_p2}, it holds for $x=3$ and $n\geq 10$. Since $p_2(x)$ is convex for $x \in [2,3]$, the property holds for $x \in [2,3]$.
\end{proof}

Now we are ready to prove the main lemma of this section.
\begin{lemma}
	\label{lem:f_bigger_than_p_1h_1+p_2h_2}
	It holds that
	\begin{equation}
		f(x):=x-2 - p_1(x)h_1(x)-p_2(x)h_2(x) \geq 0 \qquad \text{ for } x \in [0,n].
	\end{equation}
\end{lemma}
\begin{proof}
	Note that $h_1(x),~h_2(x)\leq 0$ for $x \in [0,1]$. For all $x \in[0,\frac{1}{2}]$, $p_1(x)$ is decreasing and $h_1(x)$ is increasing in $x$. Thus, by Property~\eqref{prop:1_p1}, for $x \in[0,\frac{1}{2}]$, 
	$$
	f(x)\geq x-2-p_1(x)h_1(x) \geq  -2-p_1(1/2)h_1(1/2)
	\geq -2+4 \cdot \frac{1}{2}\geq 0.
	$$
	For $x \in [\frac{1}{2},1]$, both $p_2(x)$ and $h_2(x)$ are decreasing. Thus, for $x \in [\frac{1}{2},1]$ and by Property~\eqref{prop:1_p2},
	$$
	f(x) \geq x-2-p_2(x)h_2(x) \geq  -\frac{3}{2}-p_2(1)h_2(1/2)
	\geq -\frac{3}{2}+4\cdot \frac{3}{4}\geq 0.
	$$
	To prove the statement for $x \in [1,2]$, we show that for every $a \in [0,1]$, we have
	$f(2-a) \geq 0.$
	By construction, we have $f(2)=0$. Thus, the property holds for $a=0$. By Property~\eqref{prop:2_p1}, for polynomial $p_1$, we have $p_1(2-a) \leq (0.9a + 0.1)a^2$. By Properties~\eqref{prop:2_p2} and~\eqref{prop:3_p2}, for polynomial $p_2$, we have $p_2(2-a) \geq 1/2+a$. Thus,
	$$
	f(2-a) \geq -a- (0.9a+0.1)a^2(1-a)+(1/2+a) a(2-a)=a^2 ((0.9 a-1.8) a+1.4),
	$$
	which is nonnegative for $a\in [0,1]$. This proves the statement for $x \in [1,2]$.
	
	To prove the statement for $x \in [2,3]$ we show that for every $a \in [0,1]$, it holds that $f(2+a) \geq 0.$
	By Property~\eqref{prop:3_p1}, for $x \in [2,3]$ and $n\geq 2$, we get $p_1(2+a) \leq \frac{1}{4}a^2$.
	By Property~\eqref{prop:4_p2}, we get that $p_2(2+a) \leq -0.45a+\frac{1}{2}$. Thus,
	$$
	f(2+a) \geq a - \frac{1}{4}a^2(1+a) - \left( -0.45a +\frac{1}{2} \right)(2+a)a = (0.15 + 0.2 a) a^2,
	$$
	which is non-negative for $a\in [0,1]$. This proves the statement for $x \in [2,3]$.
	Finally, for $x \in [3,n]$, we have 
	$$
	f(x) \geq x-2 - \frac{1}{2n^2} (x-2)^2 (x-1) -\frac{1}{2n} x(x-2) \geq 0.
	$$
\end{proof}
\subsection{Proof of Theorem~\ref{thm:SC_SoS_rank}}
By Lemma~\ref{lem:f_bigger_than_p_1h_1+p_2h_2}, 
$$
x-2 - p_1(x)h_1(x)-p_2(x)h_2(x) \geq 0 
$$
for $x \in [0,n]$ and the degree of the polynomial on the LHS is at most $O(\sqrt{n}\log(n))$.
Thus, by Theorem~\ref{thm:Blekherman_TheSecond}, there exist SoS polynomials $s_0,s_1$ such that
$$
x-2 - p_1(x)h_1(x)-p_2(x)h_2(x) =s_0(x)+x(n-x)s_1(x).	
$$
Thus,
$$
x-2 = s_0(x)+x(n-x)s_1(x) +p_1(x)h_1(x)+p_2(x)h_2(x).	
$$
Remark~\ref{rem:univariate_to_multivaraite_SoS} and the fact that $|\mathbf{x}|(n-|\mathbf{x}|)$ has a degree 1 SoS certificate over the Boolean hypercube imply the existence of a degree $O(\sqrt{n}\log(n))$ certificate over the Boolean hypercube for the polynomial $\sum_{i=1}^n x_i-2$.

\end{document}